\documentclass[11pt]{article} 
\usepackage{amssymb,latexsym,amsmath,amsbsy,amsthm}
\usepackage{graphicx} 
\usepackage{cite}
\usepackage[hidelinks]{hyperref} 

\allowdisplaybreaks 

\newtheorem{theo}{Theorem}
\newtheorem{defi}[theo]{Definition}

\newtheorem{rema}[theo]{Remark}
\newtheorem{prop}[theo]{Proposition}

\makeatletter\usepackage{microtype}\g@addto@macro\@verbatim{\microtypesetup{act‌​ivate=false}}\makeatother

\renewcommand{\atop}[2]{\genfrac{}{}{0pt}{}{#1}{#2}}
\newcommand{\ii}{\mathrm{i}}
\newcommand{\mZ}{\mathbb{Z}}
\newcommand{\mE}{\mathbb{E}}
\newcommand{\ux}{\underline{x}}
\newcommand{\uD}{\underline{D}}
\newcommand{\cD}{\mathcal{D}}
\newcommand{\cP}{\mathcal{P}}
\numberwithin{equation}{section}

\begin{document}
\begin{center}
	\noindent {\Large \bf The total angular momentum algebra related to\\ the \texorpdfstring{$\mathrm{S}_3$}{S3} Dunkl Dirac equation
	} \\[5mm]
	{\bf Hendrik De Bie${}^{1}$, Roy Oste${}^{2}$, Joris Van der Jeugt${}^{3}$}\\[3mm]
	${}^{1}$Department of Mathematical Analysis, Faculty of Engineering and Architecture,\\ Ghent University, Krijgslaan 281-S8, 9000 Gent, Belgium\\[3mm]
	${}^{2,3}$Department of Applied Mathematics, Computer Science and Statistics, Faculty of Sciences, Ghent University, Krijgslaan 281-S9, 9000 Gent, Belgium\\
\end{center}
E-mail: {\tt Hendrik.DeBie@UGent.be}; {\tt Roy.Oste@UGent.be}; {\tt Joris.VanderJeugt@UGent.be}


\begin{abstract}	
	We consider the symmetry algebra generated by the total angular momentum operators, appearing as constants of motion of the $\mathrm{S}_3$ Dunkl Dirac equation. 
	The latter is a deformation of the Dirac equation by means of Dunkl operators, in our case  
	associated to the root system $A_2$, with corresponding Weyl group  $\mathrm{S}_3$, the symmetric group on three elements.  
	The explicit form of the symmetry algebra in this case is a one-parameter deformation of the classical total angular momentum algebra $\mathfrak{so}(3)$, incorporating elements of $\mathrm{S}_3$. 
	This was obtained using recent results on the symmetry algebra for a class of Dirac operators, containing in particular the Dirac-Dunkl operator for arbitrary root system. 
	For this symmetry algebra, we classify all finite-dimensional, irreducible representations and determine the conditions for the representations to be unitarizable. The class of unitary irreducible representations admits a natural realization acting on a representation space of eigenfunctions of the Dirac Hamiltonian. Using a Cauchy-Kowalevsky extension theorem we obtain explicit expressions for these eigenfunctions in terms of Jacobi polynomials.
\end{abstract}

\section{Introduction}
\label{sec:1}

Our aim is to study the symmetry algebra generated by the total angular momentum operators, as constants of motion of the Dirac equation, when modified by means of Dunkl operators. We will explain in short its context. 
The Dirac equation can be written as 
\begin{equation}\label{DiracEQ}
\ii \hbar \frac{\partial}{\partial t}\psi(x,t) = H \psi(x,t)
\end{equation}
with the Dirac Hamiltonian for a free particle having the form
\begin{equation}\label{DiracHam}
H =  c \sum_{n=1}^3 e_n p_n + e_0 m c^2\rlap{\,.}
\end{equation}
Here, $m$ is the rest mass, $c$ the speed of light, $\hbar$ the reduced Planck constant, and the components of the momentum operator in the coordinate representation are given by
\begin{equation}\label{coordrep}
	p_j = \frac{\hbar}{\ii} \frac{\partial}{\partial x_j}, \qquad j=1,2,3\rlap{\,.}
\end{equation}
For the following, we will adopt the natural convention where $\hbar = c =1$. 
The entities $e_0,e_1,e_2,e_3$ are assumed to satisfy
the anticommutation relations $\{e_i,e_j\} = e_ie_j+e_je_i = 2\delta_{ij}$ for $i,j\in\{0,1,2,3\}$. 
Though usually realized by $4\times4$ matrices, for our purposes, and with higher dimensional generalizations in mind, it suffices to consider $e_0,e_1,e_2,e_3$ as abstract generators of a Clifford algebra. 
Accordingly, the wave function $\psi$ belongs to an appropriate spinor representation space (having four components for the $4\times4$ matrix realization).

Multiplying both sides of equation~\eqref{DiracEQ} by $e_0$ (on the left), we obtain the equivalent form as an eigenvalue equation for the spacetime Dirac operator. Indeed, 
defining $\gamma^j = e_0e_j$ for $j=1,2,3$ and $\gamma^0 = e_0$, one arrives at the anticommutation relations of the Dirac algebra, $\{\gamma^{\mu},\gamma^{\nu}\} = 2 \eta^{\mu\nu}$ for $\mu,\nu\in\{0,1,2,3\}$ with $\eta^{\mu\nu}$ the Minkowski metric (to obtain the opposite sign, append a factor $\ii$ in the definition of $\gamma^{\mu}$). Up to a term proportional to the mass $m$, the Dirac Hamiltonian~\eqref{DiracHam} consists of another Dirac operator, associated to three-dimensional Euclidean space instead of four-dimensional spacetime, and which squares to the Laplace operator on Euclidean space.

In the non-relativistic setting, the Hamiltonian for a free particle is given by
\begin{equation}\label{LaplaceHam}
H_{nr} = \frac{1}{2m} \sum_{n=1}^3  p_n^2 \rlap{\,,}
\end{equation}
which is proportional to the Laplace operator on 3D Euclidean space.
In this case, the angular momentum is a constant of motion. For 
$
L_{ij} = x_ip_j -x_jp_i 
$ with  $i,j\in\{1,2,3\}$, in the Heisenberg picture we have
\[
 i\hbar \frac{\mathrm{d}}{\mathrm{d}t} L_{ij} = [H_{nr},L_{ij}] = 0 \rlap{\,.}
\]
In fact, this holds for any system with a spherically symmetric potential. 
The symmetries $L_{23},L_{31},L_{12}$ of $H_{nr}$ are seen to satisfy the commutation relations of the angular momentum algebra, that is the Lie algebra $\mathfrak{so}(3)$. 

For the Dirac Hamiltonian~\eqref{DiracHam}, the angular momentum generators $L_{ij}$ no longer commute with $H$. Instead, the total angular momentum is conserved, taking into account also spin. Indeed, in this case we have $[H,J_{ij}]=0$ for  
\begin{equation} \label{LS}
J_{ij} = L_{ij} + S_{ij}, \qquad S_{ij} = \frac{\hbar}{2\ii} e_ie_j \quad\text{with } i,j\in\{1,2,3\}\rlap{\,.}
\end{equation}
(Note that when $e_1,e_2,e_3$ are represented in terms of Pauli matrices, one has $e_xe_y = \ii e_z$ for $(x,y,z)$ a cyclic permutation of $\{1,2,3\}$, though this does not hold in the abstract setting.) 
The symmetry algebra of $H$ generated by $J_{23},J_{31},J_{12}$, what one can call the ``total angular momentum algebra'',
is again seen to be the Lie algebra $\mathfrak{so}(3)$. 

The former ``classical'' scenarios can be generalized by means of Dunkl operators~\cite{1989_Dunkl_TransAmerMathSoc_311_167,2003_Rosler}, a generalization of partial derivatives in the form of differential-difference operators associated to a root system, and invariant under its Weyl group $G$. 
These Dunkl operators retain a desirable commutative property, but allow for non-local effects through reflection terms. 
They have seen numerous applications since their introduction in for instance physical models involving reflections \cite{Genest1,Genest2,GLV, Graczyk,OVdJ2,2003_Rosler,2015_Genest&Vinet&Zhedanov_CommMathPhys_336_243,RoslerVoit}.

Recently, Feigin and Hakobyan~\cite{Feigin} considered a deformation of the quantum angular momentum generators by means of Dunkl operators, in the context of Calogero-Moser systems. In fact, using Dunkl operators instead of the coordinate representation momentum operators~\eqref{coordrep} for the free particle Hamiltonian~\eqref{LaplaceHam}, one arrives at the Calogero-Moser Hamiltonian in harmonic confinement~\cite{Feigin}. 
The algebraic relations for the symmetry algebra of this Hamiltonian were determined and seen to constitute a deformation of the standard angular momentum algebra~\cite{Feigin}.  

Our aim is now to study the deformation of the total angular momentum algebra obtained through the use of Dunkl operators as momentum operators in the Dirac Hamiltonian~\eqref{DiracHam}.  
This can be interpreted as the addition of a potential term to $H$, whose form depends on the choice of reflection group or root system.


Given the 3D Dirac operator appearing in $H$, a useful aid in this study will be some recent results on symmetries of Dirac-Dunkl operators, in arbitrary dimension $N$~\cite{Oste}. 
The Dirac-Dunkl operator is obtained by replacing the derivatives $\partial/\partial x_j$ by Dunkl derivatives $\cD_j$ in the Dirac operator 
\begin{equation} \label{DiracOP}
\sum_{j=1}^N e_j \frac{\partial}{\partial x_j}\quad \longrightarrow \quad\sum_{j=1}^N e_j \cD_j \rlap{\,,}
\end{equation} 
and has appeared also in other contexts, e.g.~\cite{2012_DeBie&Orsted&Somberg&Soucek_TransAmerMathSoc_364_3875,2009_Orsted&Somberg&Soucek_AdvApplCliffAlg_19_403}. 
In $N$-dimensional Euclidean space, the system of Dunkl operators (and thus the reflection terms) depends on the choice of a (reduced) root system, or explicitly on the generators of the underlying reflection group $G$. 

In recent work~\cite{DeBie&Genest&Vinet-2016}, the symmetry algebra of the Dirac-Dunkl operator for $N=3$ and $G=(\mathbb{Z}_2)^3$ (and of the associated Dirac equation on the two-sphere) was identified as the so-called Bannai-Ito algebra~\cite{2012_Tsujimoto&Vinet&Zhedanov_AdvMath_229_2123}. 
This lead to the construction of representations of the Bannai-Ito algebra using the actions of the Dunkl operators.  
Moreover, by moving up in dimension a higher rank version of the Bannai-Ito algebra was postulated as the symmetry algebra of the $(\mathbb{Z}_2)^N$ Dirac-Dunkl operator~\cite{DeBie&Genest&Vinet-2016-2}. 
For a recent overview of the Bannai-Ito algebra and its applications, we refer the reader to ref.~\cite{GLV}. 

These results inspired our investigation into the Dunkl version of the Dirac operator for another reflection group, the symmetric group on three elements $\mathrm{S}_3$, associated to the root system $A_2$. 
Doing so, this provided a stepping stone towards the determination of the symmetry algebra for a bigger class of generalized Laplace and Dirac operators in general dimension $N$, in the framework of Wigner systems~\cite{Oste}. 
These results contain in particular the Dunkl versions for arbitrary root system (that is, for arbitrary $N$ and general $G$), though they hold for a more general class of abstract Dirac operators.
Armed with these new tools we now return to the three-dimensional case with the objective of finding representations, and explicit realizations, of these abstract symmetry algebras. 

In three dimensions, the symmetry algebra of such a Dirac operator forms an extension of the classical total angular momentum algebra, the Lie algebra $\mathfrak{so}(3)$. 
The algebraic relations were obtained in abstract form in~\cite{Oste} and are given by
\begin{align}
[ O_{23}, O_{12} ] 
& =  O_{31} +\{O_{123},O_2\} + [O_3,O_1] \notag  \\
[  O_{31}, O_{23} ] 
& = O_{12}  +\{O_{123},O_3\} + [O_1,O_2] \label{Osymalg} \\
[ O_{12},O_{31}  ] 
& = O_{23} +\{O_{123},O_1\} + [O_2,O_3]\notag
\end{align}
where $[A,B] = AB-BA$ and $\{A,B\}=AB+BA$ are respectively the commutator and the anticommutator of $A$ and $B$. 
This algebra, which we will denote by $\mathcal{O}_3$, is generated by seven generally non-trivial elements: $O_1,O_2,O_3,O_{12},O_{23},O_{31},O_{123}$. 
The general expressions of these symmetries are given in~\cite[formulas (3.8) and (3.10)]{Oste}. 
For the classical Dirac operator in terms of the standard partial derivatives, the one-index symmetries $O_1,O_2,O_3$ were seen to be identically zero and thus in this case the commutation relations~\eqref{Osymalg} indeed reduce to those of the Lie algebra $\mathfrak{so}(3)$. 
For other types of Dirac operators, the relations~\eqref{Osymalg} form an extension of $\mathfrak{so}(3)$ whose nature depends on the explicit form of the one-index symmetries $O_1,O_2,O_3$ in particular. 
When dealing with Dunkl operators, the choice of root system and associated reflection group $G$ is what determines the structure and the explicit form of the one-index symmetries $O_1,O_2,O_3$ and as a consequence also of the symmetry algebra, as seen from the right-hand side of the algebraic relations~\eqref{Osymalg}.

For the current paper, this is the root system $A_2$ with Weyl group $G=\mathrm{S}_3$, the symmetric group on three elements. 
We will use the notation $\mathcal{O}\mathrm{S}_3$ to denote the specific form the abstract algebra $\mathcal{O}_3$ takes on in the case of the $\mathrm{S}_3$ Dirac-Dunkl operator. 
The explicit relations of the algebra generators are given in~\eqref{OijOki}. 
From these expressions it is clear that one can speak of a one-parameter deformation of the total angular momentum algebra $\mathfrak{so}(3)$ incorporating the symmetric group $\mathrm{S}_3$.  
When the deformation parameter $\kappa$ is set to zero, one recovers the ordinary $\mathfrak{so}(3)$ algebra, as the Dunkl operators then reduce to regular partial derivatives. 
For non-zero $\kappa$, the algebra relations~\eqref{OijOki} provide an interesting and exciting new structure: a deformation of $\mathfrak{so}(3)$ by means of elements of the $\mathrm{S}_3$ group algebra.
This algebra $\mathcal{O}\mathrm{S}_3$ is the main object of study in this paper.
In particular, we shall classify (finite-dimensional irreducible) representations of this algebra, and provide an explicit realization of a class of representations in terms of orthogonal polynomials.

We briefly elaborate upon other cases which have been considered. 
For the case of $(\mathbb{Z}_2)^3$ Dirac-Dunkl operator, the commutators in the left-hand side of the algebraic relations~\eqref{Osymalg} become anticommutators through the use of commuting involutions present in the reflection group. 
This yields the Bannai-Ito algebra as the symmetry algebra~\cite{DeBie&Genest&Vinet-2016}. 
In more recent work~\cite{GLV}, the root system of type $B_3$ was also considered to define extensions of the Bannai-Ito algebra. 
A crucial ingredient in this paper is again the existence of commuting involutions in the reflection group. 
The lack of such involutions characterizes the case at hand of the symmetric group, and its importance as a reference work for future investigations in higher dimensions. 

In the subsequent section, we go over the definitions and notions required to introduce the Dirac-Dunkl operator related to $\mathrm{S}_3$. 
In section~\ref{sec:3}, we elaborate on the explicit expressions of the symmetries of this operator and give the algebraic relations~\eqref{Osymalg} for this specific case. 
In section~\ref{sec:4}, we construct a form of ladder operators and use them to classify all finite-dimensional irreducible representations of the symmetry algebra in abstract form. 
In the last section we determine explicit expressions for wavefunctions which form a unitary irreducible representation of the symmetry algebra, as realized in the framework of Dunkl operators.

\section{The \texorpdfstring{$\mathrm{S}_3$}{S3} Dunkl Dirac Hamiltonian}
\label{sec:2}

We begin by introducing the concepts needed to define the Dunkl operators we will use to deform the Dirac Hamiltonian~\eqref{DiracHam}.
We consider three-dimensional Euclidean space $\mathbb{R}^3$ with coordinates $x_1,x_2,x_3$. The symmetric group $\mathrm{S}_3$ is generated by the transpositions $g_{12},g_{23},g_{31}$ which act on functions on $\mathbb{R}^3$ in the following way
\[
g_{12} f(x_1,x_2,x_3) =  f(x_2,x_1,x_3) ,\quad  g_{23} f(x_1,x_2,x_3) =  f(x_1,x_3,x_2) , \quad 
g_{31} f(x_1,x_2,x_3) =  f(x_3,x_2,x_1)\rlap{\,.}
\] 
Denoting the two even elements by $g_{123} = g_{12}g_{23} = g_{31}g_{12} = g_{23}g_{31}$ and $g_{321} = g_{23}g_{12} = g_{12}g_{31} = g_{31}g_{23}$, the six elements of $\mathrm{S}_3$ are $ \{ 1, g_{12},g_{23},g_{31},g_{123},g_{321}\}$. For convenience we give the multiplication table of $\mathrm{S}_3$ in Table~\ref{tab:1}.
\begin{table}[h!]
	\caption{Multiplication table of $\mathrm{S}_3$.}
	\label{tab:1}       
	\[
	\begin{array}{ccccccc}
	\quad\nearrow\quad & 1 & g_{12} & g_{23} & g_{31} & g_{123} & g_{321} \\ \hline
	1  & 1 & g_{12} & g_{23} & g_{31} & g_{123} & g_{321} \\
	g_{12} & g_{12}& 1  & g_{123} & g_{321} & g_{23} & g_{31} \\
	g_{23} & g_{23}  & g_{321}& 1  & g_{123} &  g_{31} & g_{12}  \\
	g_{31}  & g_{31} & g_{123} & g_{321} & 1 & g_{12} & g_{23} \\
	g_{123} & g_{123} & g_{31} & g_{12} & g_{23} & g_{321}  & 1\\
	g_{321} & g_{321} & g_{23} & g_{31}  & g_{12} & 1 & g_{123}
	\end{array}
	\]
\end{table}

The symmetric group $\mathrm{S}_3$ arises as the Weyl group of the root system $A_2$. 
The associated Dunkl operators 
are explicitly given by~\cite{1989_Dunkl_TransAmerMathSoc_311_167,2003_Rosler}
\begin{equation}\label{Dunkl}
\begin{split}
\cD_1 & = \partial_{x_1} + \kappa\left( \frac{1-g_{12}}{x_1-x_2} +  \frac{1-g_{13}}{x_1-x_3} \right),
\qquad
\cD_2 = \partial_{x_2} + \kappa\left( \frac{1-g_{12}}{x_2-x_1} +  \frac{1-g_{23}}{x_2-x_3} \right),\\
\cD_3 & = \partial_{x_3} + \kappa\left( \frac{1-g_{31}}{x_3-x_1} +  \frac{1-g_{23}}{x_3-x_2} \right)\rlap{\,.}
\end{split}
\end{equation}
Here the parameter $\kappa$ denotes the value of the multiplicity function on the single conjugacy class all transpositions of the symmetric group share. This multiplicity function is usually taken to be real and non-negative, in order to have some favorable properties such as intertwining operators~\cite{Dunkl2}. 
Now, the property that makes these generalizations of partial derivatives so special is that they commute with one another, $[\cD_i,\cD_j]=0$ for $i,j\in\{1,2,3\}$. 
Moreover, for $i,j,k$ a cyclic permutation of $1,2,3$, the action of $\mathrm{S}_3$ on the Dunkl operators is simply given by   
\[
g_{ij} \cD_i = \cD_j g_{ij}, \quad g_{ij} \cD_j = \cD_i g_{ij}, \quad g_{ij} \cD_k = \cD_k g_{ij}\rlap{\,.}
\]
The commutation relations with the coordinate variables are easily shown to be 
\begin{equation}\label{comrel}
[\cD_i,x_j] = \cD_ix_j-x_j\cD_i = \begin{cases} 1+\kappa\displaystyle\sum_{k\neq i}g_{ik} & i=j \\ -\kappa g_{ij} & i\neq j \end{cases}
\end{equation}
for $i,j,k\in\{1,2,3\}$. Note that when $\kappa=0$ these reduce to the standard relations as the Dunkl operators then reduce to ordinary partial derivatives. 

The Laplace-Dunkl operator is given by the sum of the squares of the coordinate Dunkl operators
\begin{equation}\label{Laplace}
\Delta = (\cD_1)^2+(\cD_2)^2+(\cD_3)^2\rlap{\,,}
\end{equation}
which is obviously invariant under the action of $\mathrm{S}_3$. It is independent of the choice of orthonormal basis of $\mathbb{R}^3$. 
In this setting, the Dirac-Dunkl operator $\underline{D}$ is defined as a square root of the Dunkl Laplacian as follows:
\begin{equation}\label{DiracDunklOp}
\underline{D}= e_1  \cD_1 +  e_2 \cD_2 + e_3  \cD_3\rlap{\,,}
\end{equation}
where $e_1,e_2,e_3$ generate the three-dimensional Euclidean Clifford algebra and satisfy the anticommutation relations $\{e_i,e_j\} = 2\delta_{ij}$ for $i,j\in\{1,2,3\}$. 
The three-dimensional Euclidean Clifford algebra can be realized by means of the well-known Pauli matrices. 
For the first part of this paper, we will work with abstract Clifford elements $e_1,e_2,e_3$. We will use the Pauli matrices for the explicit construction of representation spaces in Section~\ref{sec:5}. 

The deformation of the Dirac Hamiltonian~\eqref{DiracHam} by means of Dunkl operators is given by
\begin{equation}\label{DiracDunklHam}
H_{\kappa}=      \sum_{n=1}^3 e_n  \cD_n + e_0 m =   \sum_{n=1}^3 e_n  \frac{\partial}{\partial x_n} + e_0 m   +  \kappa \sum_{i=1}^3 e_i \left( \frac{1-g_{ij}}{x_i-x_j} +  \frac{1-g_{ik}}{x_i-x_k} \right)\rlap{,}
\end{equation}
where in the last summation $i,j,k$ is a cyclic permutation of $1,2,3$. Note that for notational convenience $\hbar = c = 1$, and we have left out the imaginary unit $\ii$ which is used to make the momentum operators self-adjoint. The self-adjointness is easily recovered by accompanying every Dunkl operator in the following with a factor $1/\ii$ (or $\hbar/\ii$ when $\hbar \neq 1$). We see that the use of Dunkl operators corresponds to the addition of a potential term to the Hamiltonian $H$. 

Before moving on to the symmetries, we will briefly elaborate upon an algebraic structure naturally related to the Dirac operator. 
Together with the vector variable $\underline{x}=e_1  x_1 +  e_2 x_2 + e_3  x_3$, the operator $\underline{D}$ generates a realization of the $\mathfrak{osp}(1|2)$ Lie superalgebra~\cite{2012_DeBie&Orsted&Somberg&Soucek_TransAmerMathSoc_364_3875,MR1773773}, with relations
\begin{equation}\label{osp12}
[ \{ \underline{D},\underline{x}\}, \underline{D}] = -2 \underline{D},\qquad [ \{ \underline{D},\underline{x}\}, \underline{x}] = 2 \underline{x}.
\end{equation}
Here, the anticommutator $\{ \underline{D},\underline{x}\}$ can be written in terms of the Euler operator $\mE$ as follows
\begin{equation}\label{Euler}
\{ \underline{D},\underline{x}\}= 2\mathbb{E}+3+6\kappa\rlap{\,,}\qquad \mathbb{E} = x_1 \partial_{x_1} + x_2 \partial_{x_2} + x_3 \partial_{x_3}\rlap{\,.}
\end{equation} 
The so-called Scasimir operator~\cite{MR1773773} of this $\mathfrak{osp}(1|2)$ realization, given by 
\begin{equation} \label{Scasi}
\Gamma +1 =  \frac12\left([\underline{D},\underline{x}]-1\right) \rlap{,}
\end{equation}
satisfies $\{\Gamma+1,\uD\} = 0$ and $\{\Gamma+1,\ux\}=0$, while commuting with the even elements of $\mathfrak{osp}(1|2)$. It squares to the Casimir element, which generates the center of $\mathfrak{osp}(1|2)$.
The notation $\Gamma$ refers to its appearance in the expression for the Dirac operator in spherical coordinates. This angular Dirac operator $\Gamma$ was the main object of study for the $\mathbb{Z}_2^3$ Dunkl case~\cite{DeBie&Genest&Vinet-2016}.

Working out the commutator in the right-hand side of~\eqref{Scasi} using~\eqref{comrel} we obtain
\begin{equation}\label{Gamma}
\Gamma +1 = 1 + \kappa(g_{12}+g_{23}+g_{31}) - e_1e_2L_{12} - e_2e_3L_{23}- e_3e_1L_{31} \rlap{\,.}
\end{equation}
Here, the Dunkl versions of the angular momentum operators are defined as
\begin{equation}\label{DunklAM}
L_{12} = x_1 \cD_2 - x_2 \cD_1, \quad L_{23} = x_2 \cD_3 - x_3\cD_2, \quad  L_{31} = x_3 \cD_1 - x_1 \cD_3,
\end{equation}
where we have again left out the imaginary unit for notational convenience. In the classical case, for $\kappa=0$, the expression~\eqref{Gamma} for $\Gamma$ is seen to correspond to the spin-orbit interaction $\boldsymbol{L}\cdot\boldsymbol{S}$, with the angular momentum and spin angular momentum given by
\[
\boldsymbol{L} = (L_{23} , L_{31}, L_{12}) \rlap{\,,}\qquad 
 \boldsymbol{S} = \frac{\hbar}{2} ( e_2e_3,e_3e_1,e_1e_2)\rlap{\,.}
\]

Using the property $[\Delta,x_j] = 2 \cD_j$, the Dunkl angular momentum operators~\eqref{DunklAM} are easily shown to commute with the Dunkl Laplacian~\eqref{Laplace}. It is for these operators (and generalizations thereof in dimension $N$) that the ``Dunkl angular momentum algebra'' was determined in ref.~\cite{Feigin}.
In the next section we will present the ``Dunkl total angular momentum algebra.''

\section{Symmetry algebra of the \texorpdfstring{$\mathrm{S}_3$}{S3} Dunkl Dirac Hamiltonian}
\label{sec:3}

The Dirac-Dunkl operator~\eqref{DiracDunklOp} appearing in the Hamiltonian~\eqref{DiracDunklHam} is a special case of a class of generalized Dirac operators for which the symmetry algebra was obtained recently in abstract form~\cite{Oste}. This symmetry algebra, in general, is generated by elements which either commute or anticommute with the Dirac-Dunkl operator. 
The constants of motion commuting with the Dunkl Dirac Hamiltonian $H_{\kappa}$ will follow from these results. 

In three dimensions, the symmetry algebra, denoted by $\mathcal{O}_3$, is governed by the relations~\eqref{Osymalg}. 
It consists of three one-index symmetries $O_1,O_2,O_3$ and a three-index symmetry $O_{123}$ which anticommute with the Dirac operator, and three two-index symmetries $O_{12},O_{23},O_{31}$ which commute with the Dirac operator.
These two-index symmetries will play the role of total angular momentum operators. 
The case at hand is that of the reflection group $\mathrm{S}_3$, with symmetry algebra denoted by $\mathcal{O}\mathrm{S}_3$. We will elaborate upon the explicit form the symmetries and the relations~\eqref{Osymalg} take on for this case. 

For the $\mathrm{S}_3$ case, the one-index symmetries are explicitly given by~\cite[Theorem 3.6]{Oste}
\begin{equation}\label{Oi}
O_{1}  =\frac{\kappa}{\sqrt{2}} ( G_{12}-G_{31})\qquad
O_{2}  = \frac{\kappa}{\sqrt{2}} ( G_{23}-G_{12})\qquad
O_{3}  =  \frac{\kappa}{\sqrt{2}} ( G_{31}-G_{23}),
\end{equation}
where 
\[
G_{12} = \frac{1}{\sqrt2}g_{12} (e_1-e_2),\quad G_{23} =  \frac{1}{\sqrt2}g_{23} (e_2-e_3),\quad G_{31} = \frac{1}{\sqrt2}g_{31} (e_3-e_1).
\]
Note that the three one-index symmetries are not independent as $O_3 = -O_1 - O_2$, and moreover $O_1O_2O_1 = (3\kappa^2/2)O_3$. As a direct consequence of $e_0$ anticommuting with $e_1,e_2,e_3$, we see that $G_{12},G_{23},G_{31}$ and in turn $O_1,O_2,O_3$ anticommute with the Dunkl Dirac Hamiltonian $H_{\kappa}$. 

The operators $G_{12},G_{23},G_{31}$ appearing here consist of a transposition of $\mathrm{S}_3$ appended with the
Clifford element corresponding to the normed vector in the root system associated to the reflection in question (which is an element of the Pin group of the Clifford algebra). It was observed already~\cite{Oste} 
that they also anticommute with $\underline{D}$ (one easily verifies this by direct computation)
\[
\{\underline{D},G_{12}\}=0,\qquad \{\underline{D},G_{23}\}=0,\qquad \{\underline{D},G_{31}\}=0.
\]
The symmetries $G_{12},G_{23},G_{31}$ in fact generate a new copy of the symmetric group $\mathrm{S}_3$, which extends its action to affect also Clifford algebra elements, with an extra minus sign.
Indeed, we have
\begin{align*}
(G_{ij})^2 & = 1, &   G_{ij} e_iG_{ij}&  = -e_j , &  G_{ij} e_j G_{ij} &= -e_i, & G_{ij} e_kG_{ij} &= -e_k
\end{align*} 
where $(i,j,k)$ is a cyclic permutation of $\{1,2,3\}$. 
Moreover, $G_{12} G_{23} G_{12} = G_{31}$ with analogous relations for conjugation with $G_{23}$ and $G_{31}$.
The symmetries corresponding to the two even elements of $\mathrm{S}_3$ are
\[
G_{123} = G_{12}G_{23} = \frac12 g_{123}( e_1e_2+e_2e_3+e_3e_1-1) = G_{23}G_{31} = G_{31}G_{12},
\]
\[
G_{321} = G_{23}G_{12} = \frac12 g_{321}( e_2e_1+e_3e_2+e_1e_3-1) = G_{31}G_{23} = G_{12}G_{31},
\]
which both commute with the $\mathrm{S}_3$ Dirac-Dunkl operator, the element $e_0$, and hence with $H_{\kappa}$.
This gives two constants of motion, directly related to the underlying reflection group of the Dunkl operators, corresponding to the actions of cyclically permuting the coordinates $x_1,x_2,x_3$. 
Indeed, the addition of the potential term to the Hamiltonian~\eqref{DiracDunklHam} breaks the spherical symmetry, being invariant only under a subgroup of the rotation group $\mathrm{SO}(3)$.
 
The individual symmetries $G_{12},G_{23},G_{31}$ are not contained in the algebra $\mathcal{O}\mathrm{S}_3$. However, the one-index symmetries $O_1,O_2,O_3$ are built up from $G_{12},G_{23},G_{31}$, so it is useful to extend the symmetry algebra to contain also this realization of $\mathrm{S}_3$. We will denote this extension by $\mathcal{O}\mathrm{S}t_3$.

The two-index symmetries $O_{ij}$ commute with the $\mathrm{S}_3$ Dirac-Dunkl operator. These Dunkl versions of the total angular momentum operators are explicitly given by~\cite[Example 4.2.2]{Oste}
\begin{align}
O_{ij} & = L_{ij} + \frac12 e_i e_j + O_ie_j-O_je_i\notag\\
& = L_{ij} + \frac12 e_i e_j + \frac{\kappa}{\sqrt{2}} ( G_{ij}e_i-G_{jk}e_i+G_{ij}e_j-G_{ki}e_j)\label{Oij}\\
& = L_{ij} + \frac12 e_i e_j + \kappa(g_{12}+g_{23}+g_{31})e_ie_j-O_ke_1e_2e_3\notag
\end{align}
where $i,j,k$ is a cyclic permutation of $1,2,3$, $L_{ij}$ is a Dunkl angular momentum operator~\eqref{DunklAM}, and the last line follows by means of the identity
\[
O_1 e_1+O_2 e_2+O_3 e_3 = \kappa(g_{12}+g_{23}+g_{31})\,.
\]
For $\kappa=0$, they reduce to the classical total angular momentum operators~\eqref{LS} (up to multiplication by the imaginary unit $\ii$). When $\kappa$ is nonzero, they contain besides a Dunkl angular momentum term~\eqref{DunklAM} and a spin term, also a non-trivial part involving reflections and Clifford elements.

As each term of $O_{ij}$ contains an even number of Clifford algebra generators $e_1,e_2,e_3$, using again their anticommutation relations, we find that the symmetries $O_{12},O_{23},O_{31}$ indeed 
commute with the Dirac-Dunkl Hamiltonian $H_{\kappa}$.


The final symmetry is the three-index symmetry
\begin{equation}\label{O123}
O_{123} 	 =  e_1e_2e_3+   O_1 e_2e_3+O_2 e_3e_1+O_3 e_1e_2   +L_{12} e_3 +L_{23} e_1+L_{31} e_2 
\end{equation}
which anticommutes with the $\mathrm{S}_3$ Dirac-Dunkl operator.
This symmetry is equal to the Scasimir of $\mathfrak{osp}(1|2)$ given by~\eqref{Scasi}, multiplied by $e_1e_2e_3$ (as obtained already in general in ref.~\cite{Oste}): 
\begin{equation*}
O_{123} =\frac12\left( [\uD,\ux]-1\right) \, e_1e_2 e_3 = (\Gamma+1) \, e_1e_2 e_3 \rlap{\,.}
\end{equation*}
The entity $e_1e_2e_3$ satisfies $(e_1e_2e_3)^2=-1$ and acts as a pseudo-scalar in the 3D Clifford algebra generated by $e_1,e_2,e_3$. In fact, in the realization by means of the Pauli matrices, $e_1e_2e_3$ is just $\ii$ times the identity matrix. 
Because of the anticommutation relations of $e_1,e_2,e_3$, one immediately sees that $[\underline{D},e_1e_2e_3]=0$. However, $e_1e_2e_3$ anticommutes with $e_0$. The Scasimir element $\Gamma+1$ portrays the opposite behavior, anticommuting with $\uD$ and commuting with $e_0$. As a consequence, $O_{123}$ will anticommute with the Dirac-Dunkl Hamiltonian $H_{\kappa}$.

By direct computation, one readily shows that the symmetry $O_{123}$ is central in the algebra  $\mathcal{O}\mathrm{S}t_3$. Moreover, it can be written in terms of the other symmetries as follows
\[
O_{123} 	 =  -\frac12 e_1e_2e_3-   O_1 e_2e_3-O_2 e_3e_1-O_3 e_1e_2   +O_{12} e_3 +O_{31} e_2 +O_{23} e_1\rlap{\,.}
\]
Again by direct computation one finds 
\begin{align}\label{Casi}
(O_{123})^2 =\ & -\frac14 +O_1^2 +O_2^2 +O_3^2+  O_{12}^2 +  O_{23}^2 +  O_{31}^2 \\
=\ & O_{12}^2 +  O_{23}^2 +  O_{31}^2 -\frac32\kappa^2(G_{123}+G_{321}) +3\kappa^2 -\frac14 \rlap{\,,}\notag
\end{align}
which corresponds, up to a sign, to the Casimir element $(\Gamma+1)^2$ of the $\mathfrak{osp}(1|2)$ realization~\eqref{osp12}.

We have shown that the Dunkl Dirac Hamiltonian~\eqref{DiracDunklHam} admits also the symmetries~\eqref{Oi}, \eqref{Oij} and \eqref{O123} of the Dirac-Dunkl operator~\eqref{DiracDunklOp}. 
The two-index symmetries $O_{12},O_{23},O_{31}$ commute with the Hamiltonian $H_{\kappa}$ and generalize the classical total angular momentum operators~\eqref{LS} as constants of motion of the Dunkl Dirac equation. 
We now translate the algebraic relations~\eqref{Osymalg} of the symmetry algebra $\mathcal{O}_3$ for a general Dirac operator to our Dunkl framework, yielding the ``Dunkl total angular momentum algebra.'' 
\begin{theo}
	The algebra $\mathcal{O}\mathrm{S}t_3$ generated by the symmetries $G_{12},G_{23},G_{31}$ and $O_{12},O_{23},O_{31},O_{123}$
	is governed by the following relations:
	\begin{itemize}
		\item $O_{123}$ commutes with the other symmetries,
		\item $G_{12},G_{23},G_{31}$ generate a copy of $\mathrm{S}_3$ and act on the indices of $O_{12},O_{23},O_{31}$ by an $\mathrm{S}_3$ action with minus sign, i.e. 
		\begin{equation}\label{GO}
		G_{12}O_{12} = -O_{12}G_{12},\quad 	G_{12}O_{23} = -O_{31}G_{12},\quad 	G_{12}O_{31} = -O_{23}G_{12},
		\end{equation}
		and analogous actions of $G_{23}$ and $G_{31}$,
		\item the commutation relations 
		\begin{align}
		[ O_{12},O_{31}  ] 
		& = O_{23} +\sqrt{2}\kappa O_{123} (G_{12} -  G_{31} ) +  \frac32\kappa^2 (G_{123} - G_{321})\notag\\
		[ O_{23}, O_{12} ] 
		& =  O_{31} +\sqrt{2}\kappa O_{123} ( G_{23} -  G_{12}   ) +  \frac32\kappa^2 (G_{123} - G_{321}) \label{OijOki}  \\
		[  O_{31}, O_{23} ] 
		& = O_{12} +\sqrt{2}\kappa O_{123}(G_{31} -  G_{23}  )+  \frac32\kappa^2 (G_{123} - G_{321})  \notag
		.
		\end{align}
	\end{itemize}
	where $\kappa$ is a scalar factor.
\end{theo}
\begin{proof}
	This follows immediately from~\eqref{Osymalg}, the explicit expressions~\eqref{Oi} and
	\[
	[O_1,O_2]  = \frac32\kappa^2 (G_{123} - G_{321})= [O_2,O_3] = [O_3,O_1]\,.
	\]
\end{proof}

Note that for $\kappa = 0$ the commutation relations~\eqref{OijOki} reduce to the well-known relations of the Lie algebra $\mathfrak{so}(3)$, the classical total angular momentum algebra.
For the sequel we will consider $\kappa$ to be non-zero.

\section{Representations}
\label{sec:4}

Both from a purely mathematical point of view and because of 
their potential use in constructing physical models, we are interested in determining all finite dimensional irreducible (unitary) representations of the algebra $\mathcal{O}\mathrm{S}t_3$ in abstract form. 
We will build up irreducible representations starting from a mutual eigenvector of a set of commuting operators. 
Contrary to the classical case where one generally uses an eigenbasis for the $z$ component of the total angular momentum, corresponding to $J_{12}$ or $O_{12}$ in our notation, it will be helpful to incorporate the $\mathrm{S}_3$ structure for the Dunkl case. 
From the relations~\eqref{GO}, we see that the linear combination $O_{12} +  O_{23}+  O_{31}$ anticommutes with $G_{12},G_{23},G_{31}$ and thus commutes with the even elements $G_{123}$ and $G_{321}$ of $\mathrm{S}_3$. 
Now, in order to construct irreducible representations, we will use a form of ladder operators. Hereto, we start by defining some auxiliary operators.

\begin{defi}\label{def1}
	Say $\omega = e^{2\pi \ii/3}$, so 
	\[\omega =  -\frac12+\ii\frac{ \sqrt3}{2}\,,\qquad \omega^2  =  -\frac12-\ii\frac{ \sqrt3}{2}=\overline{\omega}=\omega^{-1}\,,\qquad  \omega^3=1\,.
	\]	We define the following linear combinations in the algebra  $\mathcal{O}\mathrm{S}t_3$, with inverse relations on the right,
	\begin{align}
	O_0 & = \frac{-\ii}{\sqrt3}(O_{12} +  O_{23}+  O_{31}), & 
	O_{12} & =   \frac{\ii}{\sqrt{6}} ( \sqrt{2} O_0 +  O_+  + O_-), \notag\\ \label{O0pm}
	O_{+} &=  -\ii\sqrt{\frac{2}{3}} (O_{12} +   \omega   O_{23}+  \omega^2 O_{31}), &
	O_{23} & = \frac{\ii}{\sqrt{6}} ( \sqrt{2} O_0 + \omega^2 O_+  +\omega\, O_-), \\
	O_{-} & =  -\ii\sqrt{\frac{2}{3}} (O_{12} + \omega^2 O_{23}+  \omega   O_{31}), &
	O_{31} & = \frac{\ii}{\sqrt{6}} (\sqrt{2}O_0+\omega\,O_+ +\omega^2 O_-)\rlap{\,.}\notag
	\end{align}
	We also define a set of linear combinations of $G_{12},G_{23},G_{31}$
	\begin{align}\label{Npm}
	N_{+}  &= G_{12} + \omega G_{23} +  \omega^2 G_{31}\,,&
	N_{-} & = G_{12} + \omega^2 G_{23} +  \omega G_{31}\,.
	\end{align}
\end{defi}
Note that $N_+$ and $N_-$ generate the same subset of the group algebra $\mathbb{C}\mathrm{S}_3$ as $O_1,O_2,O_3$ do. The addition of $G_{12}$ yields the full $\mathrm{S}_3$ realization. 

\begin{prop}\label{Result1}
	The elements of the algebra  $\mathcal{O}\mathrm{S}t_3$ defined in Definition~\ref{def1} satisfy the relations
	\begin{align}\label{O0Opm}
	[O_0 , O_{\pm} ] = \ & \pm  O_{\pm} +2\kappa O_{123} N_{\pm} \\
	[O_+,O_-] = \ &  \ \, 2\, O_0\, +\kappa^2 [N_+,N_-] \label{OpOm}
	\end{align}
	where $[N_+,N_-] = -\ii\, 3\sqrt{3} (G_{123}-G_{321})$. 
	
	Moreover, the elements $N_{\pm}$ are nilpotent, that is $N_{\pm}^2 =0$, and satisfy
	\begin{equation}\label{NpNm2}
	(N_{\pm}N_{\mp})^2 = 9N_{\pm}N_{\mp}\rlap{\,.}
	\end{equation}
	The interaction with $O_0,O_+,O_-$ is as follows
	\begin{equation}\label{NO}
	N_{\pm} O_0= -O_0N_{\pm}\rlap{\,,}\qquad N_{\pm} O_{\pm}= -O_{\mp}N_{\mp}\rlap{\,,}\qquad N_{\pm} O_{\mp}= -O_{\pm}N_{\mp}\rlap{\,.}
	\end{equation}
	Finally, the square~\eqref{Casi} can be rewritten in the following forms
	\begin{align}\label{Cas}
	(O_{123})^2 =  & -O_{0}^2 -  \frac12\{O_{+} ,  O_{-}\} + \kappa^2\frac12\{N_{+} ,  N_{-}\}-\frac14 \notag\\
	=  &-O_{0}^2 - O_{+}  O_{-} + O_0 +\kappa^2N_+N_- -\frac14  \\
	=  &-O_{0}^2 -O_{-}O_{+}-  O_0 +\kappa^2N_-N_+  -\frac14 \, .\notag
	\end{align}
\end{prop}
\begin{proof}
	The relations~\eqref{O0Opm} and \eqref{OpOm} are proved by straightforward computations using the commutation relations~\eqref{OijOki}. 
	For the commutator of $N_+$ and $N_-$ we have
	\begin{equation}
	N_+N_-
	= (  G_{12} + \omega G_{23} +  \omega^2 G_{31}) 
	(  G_{12} + \omega^2 G_{23} +  \omega G_{31})
	= 3 + 3\omega^2 G_{123} + 3\omega G_{321} \rlap{\,,} \label{NpNm}
	\end{equation}
	while similarly
	\begin{equation}
	N_-N_+
	= (  G_{12} + \omega^2 G_{23} +  \omega G_{31})
	(  G_{12} + \omega G_{23} +  \omega^2 G_{31})
	= 3 + 3\omega G_{123} + 3\omega^2 G_{321} \rlap{\,,}   \label{NmNp}
	\end{equation}
	which leads to $[N_+,N_-] = -\ii\, 3\sqrt{3} (G_{123}-G_{321})$, and also $\{N_+,N_-\} = 6- 3 (G_{123}+G_{321})$.
	
	We illustrate the nilpotency of $N_{+}$, the result for $N_{-}$ is similar,
	\begin{align*}
	N_+^2
	=\ & \left(  G_{12} + \omega G_{23} +  \omega^2 G_{31}\right)^2\\
	=\ & 1 + \omega G_{12}  G_{23} +  \omega^2  G_{12} G_{31}
	+ \omega G_{23}G_{12} +  \omega^2 +  G_{23}G_{31}
	+  \omega^2 G_{31}G_{12} +  G_{31}G_{23} +\omega 
	\\
	=\ & 1 +\omega+  \omega^2 + ( 1 +\omega+  \omega^2) G_{123} + ( 1 +\omega+  \omega^2) G_{321} =0. 
	\end{align*}
	In the same way, starting now from the expressions~\eqref{NpNm} and \eqref{NmNp}, we obtain~\eqref{NpNm2}.
	
	The interactions in~\eqref{NO} follow immediately from 
	\begin{align}
	G_{12} O_0 &= -O_0 G_{12}& G_{23} O_0 &= - O_0 G_{23}& G_{31} O_0 &= - O_0 G_{31}\notag\\
	G_{12} O_{+} &= -O_{-} G_{12}& G_{23} O_{+} &= -\omega^2 O_{-} G_{23}& G_{31} O_{+}& = -\omega O_{-}G_{31}\label{GOpm}\\
	G_{12} O_{-} &= -O_{+} G_{12}& G_{23} O_{-} &= -\omega O_{+} G_{23}& G_{31} O_{-}& = -\omega^2 O_{+} G_{31}\,,\notag
	\end{align}
	which are direct consequences of~\eqref{GO} and the definitions~\eqref{O0pm}, \eqref{Npm}.

	Finally, the square~\eqref{Casi} is rewritten using the inverse relations~\eqref{O0pm}. 
	We find 
	\begin{align*}
	&-6 (O_{12}^2 +  O_{23}^2 +  O_{31}^2)\\
	= \ &( \sqrt{2} O_0 +  O_+  + O_-)^2 +  ( \sqrt{2} O_0 + \omega^2 O_+  +\omega O_-)^2 +  ( \sqrt{2} O_0 + \omega O_+  +\omega^2 O_-)^2 \\
	= \ &  (2+2+2) O_0^2 
	+ (1+\omega+ \omega^2) O_+^2  
	+(1+\omega^2+ \omega)  O_-^2   \\
	& + (1+\omega^2+\omega) \{\sqrt{2}O_0, O_+\} 
	+  (1+\omega+\omega^2)\{\sqrt{2}O_0,O_-\} 
	+ (1+1+1)\{O_+,O_-\} .
	\end{align*}
	The results now follow using the expression for $\{N_+,N_-\}$ and \eqref{O0Opm}.
\end{proof}

An essential ingredient for the construction and classification of representation spaces is the existence of a couple of ladder operators. 

\begin{prop}
	The elements in the algebra $\mathcal{O}\mathrm{S}t_3$
	\begin{equation}\label{Kpm}
	K_{+} = \frac12 \{O_0,O_{+}\} \qquad\qquad
	K_{-} =\frac12 \{O_0,O_{-}\}
	\end{equation}
	satisfy the relation
	\begin{align}\label{J0Jpm}
	[O_0 , K_{\pm} ] =\ & \pm  K_{\pm} \,.
	\end{align}
	Moreover, we have the factorization
	\begin{align}\label{KpKm}
	K_{+}	K_{-} & =-	\big(O_{123}^2+(O_0    -  1/2)^2
	\big)
	\big((O_0    -  1/2)^2-\kappa^2 N_{+}N_{-} \big)\\
	K_{-}K_{+} & =	-	\big(O_{123}^2+(O_0    + 1/2)^2 \big)
	\big((O_0    +  1/2)^2-\kappa^2 N_{-}N_{+} \big)\,.\label{KpKm2}
	\end{align}
\end{prop}
\begin{proof}
	We immediately find that
	\begin{align*}
	[O_0 , K_{\pm} ] =\frac12 [ O_0, \{O_0,O_{\pm}\}]
	=\frac12\{ O_0,  [O_0,O_{\pm}]\} 
	=\pm\frac12 \{ O_0,   O_{\pm}\}+\kappa\{ O_0, O_{123} N_{\pm}  \}
	=\pm K_{\pm} 
	\end{align*}	
	as $O_{123}$ commutes with $O_0$ and $N_{\pm}$, and $N_{\pm}$ anticommutes with $O_0$, see~\eqref{NO}. 
	
	The factorization of $K_{+}	K_{-}$ and $K_{-}K_{+}$ follows by long and tedious, but otherwise straightforward computations starting from the definitions~\eqref{Kpm}, and using the relations~\eqref{O0Opm},~\eqref{NO}, and the expression~\eqref{Cas}. 
\end{proof}

From~\eqref{GOpm} we find the interaction of the $\mathrm{S}_3$ realization with $K_{\pm}$ to be as follows
\begin{align}\notag
G_{12}K_{+} &=  K_{-} G_{12}\,,  &	G_{123}K_{+} &= \omega^2 K_{+} G_{123}\,,  &	G_{321}K_{+} &= \omega K_{+} G_{321}\,,\\ \label{GKpm}
G_{23}K_{+} &= \omega^2 K_{-} G_{23}\,,  & 	G_{123}K_{-} &= \omega K_{-} G_{123}\,,   &	G_{321}K_{-} &= \omega^2 K_{-} G_{321} \,,\\  \notag
G_{31}K_{+} &= \omega K_{-} G_{31}\,.   
\end{align}

Our aim is now to determine all finite-dimensional irreducible representations of $\mathcal{O}\mathrm{S}t_3$.
Hereto, let $(V,\rho_V)$ be a representation of $\mathcal{O}\mathrm{S}t_3$. From here on, we consider $V$ as an $\mathcal{O}\mathrm{S}t_3$ module by setting $G\cdot v = \rho_V(G)v$ for $G\in \mathcal{O}\mathrm{S}t_3$ and $v\in V$. 

The element $O_{123}$ commutes with all of the algebra $\mathcal{O}\mathrm{S}t_3$ so its action on an invariant subspace $V_0$ of the representation $V$ will be multiplication by a constant $\Lambda$. 
The constant $\Lambda$ will later be determined in terms of other parameters characterizing the representation. 

Following the results obtained in Proposition~\ref{Result1}, our starting point will be the element $O_0$, given by~\eqref{O0pm}, which commutes with the even $\mathrm{S}_3$ elements $G_{123}$ and $G_{321}$. 
Hence, without loss of generality, we can consider a mutual eigenvector for all these elements. 
Take $v_0\in V$ to be such an eigenvector with eigenvalue $\lambda$ for $O_0$. The eigenvalue for $G_{123}$ is restricted to the set $\{1,\omega,\omega^2\}$ as $G_{123}^3 =G_{123}G_{321}= 1 $
and if $G_{123}v_0 = \alpha v_0 $ then $G_{321}v_0 = \alpha^{-1}v_0 $. 

We will construct the $\mathcal{O}\mathrm{S}t_3$ invariant subspace containing $v_0$. 
If $V$ is irreducible this space must be either $V$ or trivial. 
The trivial case results from $v_0$ being the zero vector, so from now on we assume that $v_0$ is not the zero vector.   

If $O_0 v_0 =\lambda v_0$, then for a positive integer $k$, the vector $(K_{\pm})^k v_0$ is also an eigenvector of $O_0$. 
Indeed, using $\lbrack O_0,(K_{\pm})^k\rbrack = \pm k(K_{\pm})^k $, which follows directly from \eqref{J0Jpm}, we have 
\begin{equation}
\label{actJ0}
O_0 (K_{\pm})^k v_0 = \bigl( (K_{\pm})^k O_0 + \lbrack O_0,(K_{\pm})^k\rbrack \bigr) v_0  = K_{\pm} O_0 v_0  \pm k(K_{\pm})^k   K_{\pm} v_0 = (\lambda \pm k)  (K_{\pm})^k v_0 \rlap{\,.}
\end{equation}
The set of vectors  $\bigl\{ (K_+)^k v_0 \,\big\vert\, k \in \mathbb{N} \bigr\}$ must be linearly independent because they have distinct eigenvalues as eigenvectors of $O_0$. 
If we impose $V$ to be finite-dimensional, then $(K_+)^k v_0=0$ for some $k\in \mathbb{N}$. 
Without loss of generality we may assume that $K_+v_0 = 0$. 
Following the same reasoning, the sequence $\bigl\{ (K_-)^k v_0 \,\big\vert\, k \in \mathbb{N} \bigr\}$ must also be linearly independent and thus must terminate. 
Hence $K_-(K_-^nv_0) = 0$ for some $n \in \mathbb{N}$ and we may assume without loss of generality that $n$ is minimal in this aspect, 
i.e.~$K_-^{n}v_0 \neq 0$. 

So far, we have obtained the following vectors of the representation $V$:
\begin{equation}
\label{basis0}
\bigl\{ v_k :=(K_-)^k v_0 \,\big\vert\, k = 0,\dots,n \bigr\}\,.
\end{equation}
The space spanned by these vectors is invariant under the action of $O_0$, $G_{123},G_{321},O_{123}$ and $K_-$, with $O_0 v_k=(\lambda-k) v_k$. Recall that $G_{123} v_0 = \alpha v_0$ for $\alpha\in\{1,\omega,\omega^2\}$, or thus $\alpha = \omega^{\ell}$ for some integer $\ell$. By~\eqref{GKpm}, we then have
\[
G_{123} v_k = G_{123}(K_-)^k v_0 = \omega^{k} (K_-)^kG_{123} v_0  = \omega^{\ell+k}  v_k, \qquad G_{321}v_k =  \omega^{-\ell-k}  v_k\,.
\]
The transpositions $G_{12},G_{23},G_{31}$ all square to the identity and anticommute with $O_0$. Let $v_k^- = G_{12} v_k$, then $G_{12} v_k^- = v_k$ and $O_0 v_k^- =O_0 G_{12} v_k  =- G_{12} O_0v_k= -(\lambda-k)v_k^-$. 
Moreover, $G_{23} v_k$ and $G_{31} v_k$ must both be proportional to $v_k^-$ since the compositions $G_{123}$ and $G_{321}$ act diagonally on $v_k$, and in turn also on $v_k^-$. Indeed, we have 
\[
G_{123} v_k^- = G_{123}G_{12} v_k = G_{12}G_{321} v_k =\omega^{-\ell-k} v_k^-, \qquad G_{321}v_k^- =   \omega^{\ell+k} v_k^-\rlap{\,.}
\]
It follows from $G_{12}K_- = K_+G_{12}$ that $v_k^- = G_{12} (K_-)^k v_0 = (K_+)^k  G_{12}v_0 = (K_+)^k  v_0^-$ or thus $K_+ v_k^-=v_{k+1}^-$. 
In this way, we arrive at the following set of vectors of $V$:
\begin{equation}
\label{basis}
\mathcal{B} = \bigl\{ v_k^+ := v_k = (K_-)^k v_0 \,\big\vert\, k = 0,\dots,n \bigr\}\cup \bigl\{ v_k^- :=G_{12} v_k^+ = (K_+)^k G_{12}v_0\,\big\vert\, k = 0,\dots,n \bigr\}\rlap{\,.}
\end{equation}
All these vectors are eigenvectors of the mutually commuting elements $O_0$ and $G_{123}$:
\begin{equation}\label{actO0}
O_0 v_k^{\pm}=\pm(\lambda-k) v_k^{\pm}\rlap{\,,}
\end{equation}
for $k\in\{0,\dotsc,n\}$, while
\begin{equation}\label{actG123}
G_{123} v_k^{\pm} =  \omega^{\pm(\ell+k)} v_k^{\pm}, \qquad G_{321}v_k^{\pm} =  \omega^{\mp(\ell+k)} v_k^{\pm}\rlap{\,.}
\end{equation}
Note that the representation $V$ is characterized or labeled by $(\lambda,n,\ell)$ where $n$ is a non-negative integer and $\ell \in \mZ_3 = \mZ/3\mZ$ with $3\mathbb{Z} = \{ 3z \mid z \in \mathbb{Z}\}$ the set of multiples of 3. 

We will show that the set $\mathcal{B}$ spans the $\mathcal{O}\mathrm{S}t_3$ invariant subspace containing $v_0$, which if $V$ is irreducible must be all of $V$. Moreover, in case the $O_0$ eigenvalues are all distinct then $\mathcal{B}$ 
forms a basis for the irreducible representation $V$. 
Hereto, we determine the action of all elements on $\mathcal{B}$.

The explicit action of $G_{23}$ and $G_{31}$ follows from~\eqref{actG123} as
\begin{equation}\label{actG23}
G_{23} v_k^{\pm} =G_{12} G_{123}v_k^{\pm}=  \omega^{\pm(\ell+k)} G_{12}v_k^{\pm}=  \omega^{\pm(\ell+k)} v_k^{\mp}, \qquad G_{31} v_k^{\pm} =G_{12} G_{321}v_k^{\pm}=    \omega^{\mp(\ell+k)}  v_k^{\mp}\,,
\end{equation}
and in turn the action of $N_{\pm}$ as defined by~\eqref{Npm},
\begin{equation}\label{actNp}
N_{+} v_k^{\pm} =(G_{12}+\omega G_{23}+\omega^{2}G_{31})v_k^{\pm}= ( 1+ \omega^{ 1\pm\ell\pm k}+\omega^{-1\mp\ell\mp k})v_k^{\mp} = 3\,\mathbf{1}_{3\mathbb{Z}}(\ell+k\pm 1)v_k^{\mp}\,.
\end{equation}
where we employ the notation, 
\[
\mathbf{1}_{3\mathbb{Z}}(k) = \frac{1+\omega^k+\omega^{-k}}{3} =\begin{cases}
1 & \mbox{if} \quad k \equiv 0 \;(\bmod\; 3) \iff k \in 3\mathbb{Z} \\
0 & \mbox{if} \quad k \equiv 1,2\;(\bmod\; 3) \iff k \not\in 3\mathbb{Z}\,.
\end{cases}
\]
Similarly, we have
\begin{equation}\label{actNm}
N_{-} v_k^{\pm}  = 3\,\mathbf{1}_{3\mathbb{Z}}(\ell+k\mp 1)v_k^{\mp}\,.
\end{equation}
By~\eqref{NpNm2}, we find that the linear combinations of $G_{123}$ and $G_{321}$ denoted by $N_+N_-$ and $N_-N_+$, see~\eqref{NpNm} and \eqref{NmNp}, satisfy the polynomial equation $X^2-9X=0$. Consequently their eigenvalues are 0 and 9. 
Following~\eqref{actNp} and \eqref{actNm}, we obtain the diagonal actions 
\begin{equation}\label{actNN}
N_+N_- v_k^{\pm} =   9\, \mathbf{1}_{3\mathbb{Z}}(\ell+k\mp 1) v_k^{\pm} \,,
\qquad
\mbox{and}
\qquad
N_-N_+ v_k^{\pm} =  9\, \mathbf{1}_{3\mathbb{Z}}(\ell+k\pm 1) v_k^{\pm} \rlap{\,.}
\end{equation}

We already know that $K_- v_k^+ = v_{k+1}^+$ and $K_+ v_k^- = v_{k+1}^-$  with $v_l^{\pm} = 0$ for $l>n$. 
Using \eqref{KpKm} we find the action of $K_+$ and $K_-$ on the rest of the basis $\mathcal{B}$:
\begin{align}
K_+ v_k^+ &= K_+K_- v_{k-1}^+  = -	\big(O_{123}^2+(O_0    -  1/2)^2
\big)
\big((O_0    -  1/2)^2 -\kappa^2 N_{+}N_{-}\big)v_{k-1}^+ \notag\\
&= -	\big(\Lambda^2 + (\lambda-k+1/2)^2
\big)
\big((\lambda-k+1/2)^2 - 9\kappa^2 \mathbf{1}_{3\mathbb{Z}}(\ell+k-2)\big)v_{k-1}^+\,, \label{actKp}
\end{align}	
and similarly
\begin{align}
K_- v_k^-& = K_-K_+ v_{k-1}^-  = -\big(O_{123}^2 + (O_0    +  1/2)^2 \big)
\big((O_0    +  1/2)^2 -\kappa^2 N_{-}N_{+}\big)v_{k-1}^-\notag\\
& =- \big(\Lambda^2
+ (\lambda-k+1/2)^2  \big)
\big((\lambda-k+1/2)^2- 9\kappa^2 \mathbf{1}_{3\mathbb{Z}}(\ell+k-2)  \big)v_{k-1}^-\,.\label{actKm}
\end{align}
For ease of notation, we define the expression $A(k)$ to denote these actions, that is
\begin{equation}\label{A}
A(k) = - \big(\Lambda^2
+ (\lambda-k+1/2)^2  \big)
\big((\lambda-k+1/2)^2- 9\kappa^2 \mathbf{1}_{3\mathbb{Z}}(\ell+k+1)  \big)\,,
\end{equation}
such that 
\begin{align}
\begin{split}\label{actK}
K_+ v_{k}^+ & = A(k)v_{k-1}^+  = K_+K_- v_{k-1}^+\,,\qquad K_- v_{k}^- = A(k)v_{k-1}^- =   K_-K_+ v_{k-1}^-\\
K_+K_- v_{k}^- & =  A(k) K_+ v_{k-1}^-=  A(k)  v_{k}^-\,,\qquad
K_-K_+ v_{k}^+ =  A(k) K_- v_{k-1}^+=  A(k)  v_{k}^+ \,.
\end{split}
\end{align}

For the action of $O_+$ and $O_-$ on $\mathcal{B}$, we set out as follows. Using~\eqref{O0Opm} we have
\begin{align}
K_{\pm} = \frac12\{O_0,O_{\pm}\} = O_{\pm}O_0 + \frac12[O_0,O_{\pm}] =  O_{\pm}O_0 \pm \frac12O_{\pm} + \kappa O_{123}N_{\pm}\,.\label{KpmOpm}
\end{align}
As $K_- v_k^+ = v_{k+1}^+$ for $k\leq n-1$, we find
\[
v_{k+1}^+ = K_- v_k^+ =  O_{-}\Big(O_0  - \frac12\Big) v_k^+  + \kappa O_{123}N_{-} v_k^+ =  (\lambda -k-1/2)O_{-}v_k^+  + 3\kappa \Lambda \mathbf{1}_{3\mathbb{Z}}( \ell+k-1 ) v_k^- .
\]
Hence, for $\lambda -k-1/2\neq 0$ (we will handle the zero case after determining the possible values for $\lambda$)
\begin{equation}\label{actOmp}
O_{-}v_k^+ = \frac{1}{\lambda -k-1/2}v_{k+1}^+  
- \frac{3\kappa \Lambda}{\lambda -k-1/2} \mathbf{1}_{3\mathbb{Z}}( \ell+k-1) v_k^- \,.
\end{equation}
The action of $O_{-}$ on $v_n^+$ is consistent with~\eqref{actOmp} by letting $v_{n+1}^+=0$ as by
\[
0 = K_- v_n^+ =  O_{-}\Big(O_0  - \frac12\Big) v_n^+  + \kappa O_{123}N_{-} v_n^+ =  (\lambda -n-1/2)O_{-}v_n^+  + 3\kappa \Lambda \mathbf{1}_{3\mathbb{Z}}( \ell+n-1) v_n^- 
\]
we have, for $\lambda -n-1/2\neq 0$
\begin{equation*} 
O_{-}v_n^+ =   - \frac{3\kappa \Lambda}{\lambda -n-1/2} \mathbf{1}_{3\mathbb{Z}}(\ell+n-1) v_n^- \,.
\end{equation*}

The action~\eqref{actKm} together with~\eqref{KpmOpm} yields the action of $O_-$ on $v_k^-$. On the one hand
\[
K_- v_k^- =   O_{-}\Big(O_0  - \frac12\Big) v_k^{-}  + \kappa O_{123}N_{-} v_k^{-}  =  -(\lambda - k +1/2)O_{-}v_k^{-}  + 3\kappa \Lambda \mathbf{1}_{3\mathbb{Z}}(\ell+k+1) v_k^{+} \,,
\]
while on the other hand for $k\geq 1$ we have $ K_- v_k^- =K_-K_+ v_{k-1}^-=A(k)v_{k-1}^- $, so for $\lambda -k+1/2\neq 0$
\begin{equation}\label{actOmm}
O_{-}v_k^- = \frac{-A(k)}{\lambda - k +1/2}v_{k-1}^-  + \frac{3\kappa \Lambda}{\lambda - k +1/2} \mathbf{1}_{3\mathbb{Z}}(\ell+k+1) v_k^+ \,.
\end{equation}
For $\lambda\neq1/2$, this is consistent with the action of $O_{-}$ on $v_0^-$ by letting $v_{-1}^-=0$ as
\[
0 = K_- v_{0}^- =  O_{-}\Big(O_0  - \frac12\Big) v_{0}^-  + \kappa O_{123}N_{-} v_{0}^- =\Big(-\lambda-\frac12\Big)O_{-} v_{0}^-  +  3\kappa \Lambda \mathbf{1}_{3\mathbb{Z}}(\ell + 1) v_{0}^+\,.
\]
In a similar way we obtain the action of $O_+$ to be given by
\begin{equation}\label{actOpm}
O_{+}v_k^- = \frac{-1}{\lambda -k-1/2}v_{k+1}^-  + \frac{3\kappa \Lambda}{\lambda -k-1/2} \mathbf{1}_{3\mathbb{Z}}(\ell+k-1) v_k^+ 
\end{equation}
and since  $ K_+ v_k^+ =K_+K_- v_{k-1}^+=A(k)v_{k-1}^+ $
\begin{equation}\label{actOpp}
O_{+}v_k^+ = \frac{A(k)}{\lambda -k+1/2}v_{k-1}^+  - \frac{3\kappa \Lambda}{\lambda -k+1/2} \mathbf{1}_{3\mathbb{Z}}(\ell+k+1) v_k^- \,.
\end{equation}

The actions of all elements of the algebra $\mathcal{O}\mathrm{S}t_3$ are fixed by the four constants $ n,\lambda,\Lambda,\ell$, where $\ell$ is integer and $n$ is a positive integer. We will now examine all possible values which lead to finite irreducible representations. 
The conditions for the dimension to be finite, $K_+ v_0^+ = 0$ and $K_- v_n^+ = 0$ can be combined with the results \eqref{KpKm} and \eqref{KpKm2} of Proposition~\ref{Kpm} to find
\begin{equation*}
\begin{cases}
K_-	K_+ v_0^+ = 0 \\
K_+	K_- v_n^+ = 0
\end{cases} \iff  \begin{cases}
-	\big(O_{123}^2+(O_0    +  1/2)^2 \big)
\big((O_0    +  1/2)^2-\kappa^2 N_{-}N_{+} \big)v_0^+ = 0 \rlap{\,,}\\
-	\big(O_{123}^2+(O_0    -  1/2)^2
\big)
\big((O_0    -  1/2)^2-\kappa^2 N_{+}N_{-} \big)v_n^+ = 0\rlap{\,.}
\end{cases}
\end{equation*}
When plugging in the appropriate actions, \eqref{actO0} and \eqref{actNN}, this yields the system of equations
\begin{equation}\label{syseqs}
\begin{cases}
-\big(\Lambda^2+(\lambda + 1/2)^2
\big)
\big((\lambda + 1/2)^2  -9\kappa^2 \mathbf{1}_{3\mathbb{Z}}(\ell +1)\big) = 0 \\
-\big(\Lambda^2+(\lambda-n    -  1/2)^2 \big)
\big((\lambda-n    -  1/2)^2 -9\kappa^2 \mathbf{1}_{3\mathbb{Z}}(\ell+n-1)\big) = 0
\end{cases}
\end{equation}
to be satisfied for every set of valid values for $ n,\lambda,\Lambda,\ell$. 
We distinguish three distinct classes of solutions of~\eqref{syseqs} depending on which pair of factors are zero and on the value of $\ell$, which decides whether the function $\mathbf{1}_{3\mathbb{Z}}$ is 0 or 1. 
\begin{itemize}
	\item[(a)] $\Lambda^2+(\lambda + 1/2)^2=0$ and $(\lambda-n    -  1/2)^2 -9\kappa^2 =0$, that is when $\mathbf{1}_{3\mathbb{Z}}(\ell+n-1)=1$
	\item[(b)] $\Lambda^2+(\lambda + 1/2)^2=0$ and $(\lambda-n    -  1/2)^2=0$, that is when $\mathbf{1}_{3\mathbb{Z}}(\ell+n-1)=0$
	\item[(c)]	$\Lambda^2+(\lambda + 1/2)^2=0$ and $\Lambda^2+(\lambda-n    -  1/2)^2=0$
\end{itemize}
Note that there are two more cases we have omitted from our classification. We briefly expand on this before continuing. 
First, we have the case where $\Lambda^2+(\lambda-n    -  1/2)^2=0$ and $(\lambda + 1/2)^2-9\kappa^2 =0$, that is when $\mathbf{1}_{3\mathbb{Z}}(\ell +1)=1$. This will turn out to be equivalent to case (a) after renaming $v_k^{-}$ to $v_{n-k}^{+}$ and vice versa, as seen from the action of $O_0$. Similarly, the case where $\Lambda^2+(\lambda-n    -  1/2)^2=0$ and $(\lambda + 1/2)^2=0$ 
will be equivalent with case (b). 

We continue with the classification of all finite-dimensional irreducible representations.  Note that in all three cases $\Lambda^2+(\lambda + 1/2)^2=0$ which fixes the value $\Lambda = \varepsilon \,\ii\, (\lambda +1/2)$, up to a sign $\varepsilon=\pm 1$, and thus the action of $O_{123}$ in function of $\lambda$. This leaves a freedom in the choice of sign of $\Lambda$. In the algebra relations~\eqref{OijOki}, and \eqref{O0Opm}--\eqref{OpOm} by extension, the central element $O_{123}$ is always accompanied by a single factor $\kappa$. We note that when one permits also negative values for $\kappa$, the sign of $\Lambda$ can always be chosen such that
the product $\kappa O_{123}$ has a positive action. 

For each case we need to check whether the vectors~\eqref{basis} are independent. Since $v_0$ is a generating vector of $V$, irreducibility can be checked by verifying that for each $v_k^{\pm}$ there is an algebra element $X_k^{\pm}$ such that $X_k^{\pm}v_k^{\pm}= v_0 $. Note that $(K_+)^kv_k^+=A(k)A(k-1)\dotsm A(1) v_0$ by~\eqref{actK}, while $G_{12}(K_-)^kv_k^-=A(k)A(k-1)\dotsm A(1) G_{12}v_0^-$ and $G_{12}v_0^-=v_0^+=v_0$. Hence, in order to have an irreducible representation, the expression $A(k)$, given by~\eqref{A}, must be non-zero for $k\in\{1,\dotsc,n\}$. 
Plugging in $\Lambda^2=-(\lambda + 1/2)^2$, we find 
\begin{align}
A(k) & = - \big(-(\lambda + 1/2)^2
+ (\lambda-k+1/2)^2  \big)
\big((\lambda-k+1/2)^2- 9\kappa^2 \mathbf{1}_{3\mathbb{Z}}(\ell+k+1)  \big)\notag\\
& =  
k(2\lambda+1-k)  
\big((\lambda-k+1/2)^2- 9\kappa^2 \mathbf{1}_{3\mathbb{Z}}(\ell+k+1)  \big)
\,.\label{Al}
\end{align}
We now work out the explicit value of $\lambda$ and $\ell$ for the three cases.

\subsection{Case (a)}
\label{sec:4.1}

For the case (a) we have $\mathbf{1}_{3\mathbb{Z}}(\ell+n-1)=1$ or thus $\ell \equiv 2n+1 \;(\bmod\; 3) $ which fixes the eigenvalues of the reflections, e.g.~$G_{123} v_k^{\pm} = \omega^{\pm(  2n+1+k )} v_k^{\pm}$, see~\eqref{actG123} and \eqref{actG23}. 
Moreover, from  $(\lambda-n    -  1/2)^2 -9\kappa^2 =0$ we find $\lambda=n\pm3\kappa+1/2$. With unitary representations in mind, we first handle the case $\lambda=n+3\kappa+1/2$. 
The action of $O_0$ on $\mathcal{B}$ is now given by 
\begin{equation}\label{actO0a}
O_0 v_k^{\pm}=\pm(n-k+3\kappa+1/2) v_k^{\pm},\qquad k\in\{0,\dotsc,n\}\,.
\end{equation}
We see that for a positive parameter $\kappa$ every vector of the set $\mathcal{B}$ has a distinct eigenvalue for $O_0$, so the elements of $\mathcal{B}$ are independent. Note that $\lambda -k\pm1/2\neq 0$ for $\kappa>0$ and $k\in\{0,1,\dotsc,n\}$ so the previously determined actions of the elements of $\mathcal{O}\mathrm{S}t_3$ on $\mathcal{B}$, see e.g.~\eqref{actOmp}, are all well-defined. 
Moreover, the expression~\eqref{Al} is readily seen to be non-zero for all positive values $\kappa$ and $k\in\{1,\dotsc,n\}$. 
This shows that, for positive $\kappa$, the set $\mathcal{B}$ forms a basis for the $\mathcal{O}\mathrm{S}t_3$ invariant subspace containing $v_0$, which if $V$ is irreducible must be all of $V$. The actions of the other generators of $\mathcal{O}\mathrm{S}t_3$ are given by~\eqref{actOmp}, \eqref{actOmm}, \eqref{actOpm}, \eqref{actOpp}.

Next, we consider the other choice $\lambda=n-3\kappa+1/2$. 
The $O_0$-eigenvalues~\eqref{actO0} are not necessarily all distinct when
$
6\kappa \in \{1,2,\dotsc,2n+1\}
$. 
Moreover, the condition for irreducibility now leads to disallowed values for $\kappa$,  namely
$
6\kappa \notin \{n+2,n+3,\dotsc,2n+1\}\cup( \{1,\dotsc,n\}\cap 3\mathbb{Z})
$, while also 
$
3\kappa-1 \notin \{0,1,\dotsc,n-1\}\cap 3\mathbb{Z}
$, 
and 
$
3\kappa-2 \notin \{0,1,\dotsc,n-2\}\cap 3\mathbb{Z}
$. 
Hence $\mathcal{B}$ would form a basis for an irreducible $\mathcal{O}\mathrm{S}t_3$ representation if and only if $\kappa$ is not allowed to take on these specific values. As a consequence this choice will not lead to unitary representations for general values of $\kappa$.

Finally, note that the choice $\lambda=n-3\kappa+1/2$ with $\kappa$ positive
is equivalent to considering negative values for $\kappa$ when $\lambda=n+3\kappa+1/2$. 
For a given real value of $\kappa$, the sign accompanying $\kappa$ in $\lambda=n\pm3\kappa+1/2$ can thus always be chosen such that $\lambda$ is positive. For negative $\kappa$, the disallowed values follow immediately by replacing $\kappa$ by $-\kappa$ in the previously obtained conditions. 
These values correspond in fact to those of the $\mathrm{S}_3$ Dunkl operator singular parameter set for which no intertwining operators exist~\cite{Dunkl2,2003_Rosler,Dunkl3}.

\subsection{Case (b)}
\label{sec:4.2}

For the case (b) we have $\mathbf{1}_{3\mathbb{Z}}(\ell+n-1)=0$ or thus $\ell \not\equiv 2n+1 \;(\bmod\; 3) $. This gives two distinct options for the eigenvalues of $G_{123}$ and in turn for the actions of the other reflections.
The condition $(\lambda-n    -  1/2)^2=0$ implies that $\lambda = n+1/2$, which again yields $2n+2$ distinct $O_0$ eigenvalues 
\[
O_0 v_k^{\pm}=\pm(n-k+1/2) v_k^{\pm},\qquad k\in\{0,\dotsc,n\}\,.
\]
For the case at hand the acquired actions~\eqref{actOmp},\eqref{actOmm},\eqref{actOpm},\eqref{actOpp} do not lead to the full action of $O_-$ or $O_+$, as we would have to divide by zero. Indeed, we have $O_0  v_{n}^+ = \frac12 v_{n}^+ $ and $O_0  v_{n}^- = -\frac12 v_{n}^-$ so the denominator in~\eqref{actOmp} would become zero for $k=n$. We determine the action of $O_-$ on $v_n^+$ and $O_+$ on $v_n^-$ in another way. 
By means of relation~\eqref{O0Opm} acting on $v_n^+$ we find
\begin{align*}
(O_0O_{-}-O_{-}O_0)v_{n}^+ &=   - O_{-}v_{n}^+ + 2\kappa O_{123}N_{-}v_{n}^+ \\
\iff \ \, O_0O_{-}v_{n}^+ - \frac12 O_{-}v_{n}^+   &=- O_{-}v_{n}^+
\\
\iff \qquad \qquad\quad O_0O_{-}v_{n}^+    &= - \frac12 O_{-}v_{n}^+  \,,
\end{align*}
which implies $O_- v_{n}^+ = \beta_- v_n^-$ for some constant $\beta_-$. In the same manner we find $O_+ v_{n}^- = \beta_+ v_n^+$ for some constant $\beta_+$. 
Using the interaction of $G_{12}$ and $O_{\pm}$, see~\eqref{GOpm}, we find
\[
\beta_- v_n^- = O_- v_{n}^+ = O_- G_{12} v_{n}^- =-G_{12}O_+  v_{n}^-=-G_{12}\beta_+  v_{n}^+=-\beta_+  v_{n}^-\,,
\]
while by~\eqref{Cas} we have
\begin{align*}
\beta_+\beta_-v_n^+ = O_+O_-v_n^+ = \Big(-	(O_{123})^2-\Big(O_{0}-\frac12 \Big)^2  +\kappa^2N_+N_-  \Big)v_n^+ =- \Lambda^2 v_n^+\,.
\end{align*}
Hence $\beta_- = -\beta_+ = \pm \Lambda = \pm i (n+1)$. Note that we have an extra freedom in the choice of sign, besides the one present for the sign of $\Lambda $.

Finally, we check whether the expression~\eqref{Al} is non-zero for $k\in\{1,\dotsc,n\}$. For $\lambda = n+1/2$, only the factor  
$
(n-k+1)^2- 9\kappa^2 \mathbf{1}_{3\mathbb{Z}}(\ell+k+1) 
$ could become zero. 
Hereto, we distinguish between the two options for $\ell$. 
For $\ell \equiv 2n \;(\bmod\; 3) $ this gives the conditions
\[
(k+2)^2- 9\kappa^2   \neq 0 \qquad\mbox{for } k\in\{0,\dotsc,n-1\}\cap 3\mathbb{Z}\,,
\]
while $\ell \equiv 2n+2 \;(\bmod\; 3) $ leads to
\[
(k+1)^2- 9\kappa^2   \neq 0 \qquad\mbox{for } k\in\{0,\dotsc,n\}\cap 3\mathbb{Z}\,.
\]
This shows that $\mathcal{B}$ forms a basis for an irreducible $\mathcal{O}\mathrm{S}t_3$ representation if and only if $\kappa$ is not allowed to take on some specific values. 

\subsection{Case (c)}
\label{sec:4.3}

As $n$ is positive, the conditions $\Lambda^2+(\lambda + 1/2)^2=0$ and $\Lambda^2+(\lambda-n    -  1/2)^2=0$ lead to $\lambda = n/2$ and $\Lambda^2 = -(n+1)^2/4$. 
In this scenario, the vectors $v_k^{\pm}$ and $v_{n-k}^{\mp}$ have the same $O_0$ eigenvalue:
\[
O_0v_k^{\pm} = \pm\Big(\frac{n}{2} -k\Big)v_{k}^{\pm},\qquad O_0 v_{n-k}^{\mp} = \mp\Big(\frac{n}{2} -(n-k)\Big)v_{n-k}^{\mp}= \pm\Big(\frac{n}{2} -k\Big)v_{n-k}^{\mp}\,.
\]
The $G_{123}$ eigenvalues~\eqref{actG123} for $v_k^{\pm}$ and $v_{n-k}^{\mp}$ are given by
\[
G_{123}v_k^{\pm} = \omega^{\pm(\ell+k)}v_{k}^{\pm},\qquad G_{123} v_{n-k}^{\mp} =  \omega^{\mp(\ell+n-k)}v_{n-k}^{\mp}= \omega^{\mp(n-\ell)} \omega^{\pm(\ell+k)}v_{n-k}^{\mp}\,.
\]
Two different scenarios now occur depending on the value of $\ell$, that is whether $\ell\equiv n\;(\bmod\; 3) $ or not. We distinguish in the first place with respect to the parity of $n$. 

\subsubsection{Even \texorpdfstring{$n$}{n}}
\label{sec:4.3.1}

For $n$ an even integer, $\lambda = n/2$ is an integer so the previously determined actions of the elements of $\mathcal{O}\mathrm{S}t_3$ on $\mathcal{B}$ are all well-defined. When $\ell\equiv n\;(\bmod\; 3) $, the space generated by $v_0$ is comprised of two irreducible components and decomposes as follows. 
The vectors $v_{\lambda}^+$ and $v_{\lambda}^-$ both have $0$ as $O_0$ eigenvalue and $G_{12}v_{\lambda}^+=v_{\lambda}^-$. Hence, defining $u_0^+=v_{\lambda}^+ + v_{\lambda}^-$ and $u_0^-=v_{\lambda}^+ - v_{\lambda}^-$, we have $G_{12}u_0^+=u_0^+$ and $G_{12}u_0^-=-u_0^-$, while $O_0u_0^{\pm}=0$ and furthermore $G_{123}u_0^{\pm}=\omega^{n+\lambda}u_0^{\pm}=u_0^{\pm}$. If we now define $u_{-k}^{\pm} = (K_-)^ku_0^{\pm} $ and $u_{k}^{\pm} = (K_+)^ku_0^{\pm} $ for $k\in\{1,\dotsc,\lambda\}$, then the sets 
\begin{equation*}
\mathcal{B}^+ = \{  u_{k}^+ \mid k=-\lambda,\dotsc,0,\dotsc,\lambda\} \qquad \mathcal{B}^- = \{  u_{k}^- \mid k=-\lambda,\dotsc,0,\dotsc,\lambda\} 
\end{equation*}
each form the basis for an $\mathcal{O}\mathrm{S}t_3$ invariant subspace of dimension $n+1$. We go over the actions on these spaces. We have
$
O_0 u_k^{\pm} = k\,  u_k^{\pm}
$ and 
$G_{123}u_k^{\pm}=\omega^{- k}u_k^{\pm}$. Moreover, 
$G_{12}u_{k}^{\pm} = \pm u_{-k}^{\pm}$, while $G_{23}u_{k}^{\pm}=\pm\omega^{-k}u_{-k}^{\pm}$ and  $G_{31}u_{k}^{\pm}=\pm\omega^{k}u_{-k}^{\pm}$. For positive $k$, we have by definition $K_+u_{k}^{\pm}= u_{k+1}^{\pm}$ and $K_-u_{-k}^{\pm}= u_{-k-1}^{\pm}$. The other actions are found as follows. 
Note that for positive $k$, 
\[
u_{k}^{\pm} = (K_+)^ku_0^{\pm} = (K_+)^k(v_{\lambda}^+ \pm v_{\lambda}^- )= \prod_{l=0}^{k-1} A(\lambda - l) v_{\lambda-k}^+ \pm v_{\lambda+k}^- 
\]
and similarly
\[
u_{-k}^{\pm} = (K_-)^ku_0^{\pm} = (K_-)^k(v_{\lambda}^+ \pm v_{\lambda}^- )=  v_{\lambda+k}^+ \pm \prod_{l=0}^{k-1} A(\lambda - l)v_{\lambda-k}^- \,.
\]
Again for positive $k$, we then find
\[
K_+u_{-k}^{\pm} =K_+ K_- u_{-k+1}^{\pm} =A(\lambda+k) u_{-k+1}^{\pm}\rlap{\,,}
\qquad
K_-u_{k}^{\pm} =K_- K_+ u_{k-1}^{\pm} =A(\lambda+k)u_{k-1}^{\pm}\rlap{\,.}
\]
Here we used $A(\lambda+k) = A(\lambda-k+1)$, which is readily verified from~\eqref{Al} with $\lambda = n/2$ and $\ell\equiv n\;(\bmod\; 3) $. 

We check whether the expression~\eqref{Al} is non-zero for $k\in\{1,\dotsc,n\}$. For $\lambda = n/2$, the only factor of~\eqref{Al} with  $\ell \equiv n  \;(\bmod\; 3) $ that could become zero is
\[
\Big(\frac{n+1}{2}-k\Big)^2- 9\kappa^2 \mathbf{1}_{3\mathbb{Z}}(n+k+1) \,.
\]
This leads to the conditions
\[
\Big(k+\frac{3}{2}\Big)^2- 9\kappa^2     \neq 0 \qquad\mbox{for }  k\in\{-\lambda+2,\dotsc,\lambda+1\}\cap 3\mathbb{Z}\,,
\]
which shows that, except for specific $\kappa$ values, $\mathcal{B}^+$ and $\mathcal{B}^-$ each form the basis for an $\mathcal{O}\mathrm{S}t_3$ invariant space.

If $\ell \not\equiv n \;(\bmod\; 3) $, then $v_k^{\pm}$ and $v_{n-k}^{\mp}$ have different eigenvalues for $G_{123}$. We check whether the expression~\eqref{Al} is non-zero for $k\in\{1,\dotsc,n\}$. For $\lambda = n/2$, the only factor of~\eqref{Al} that could become zero is
\[
\Big(\frac{n+1}{2}-k\Big)^2- 9\kappa^2 \mathbf{1}_{3\mathbb{Z}}(\ell+k+1) \,.
\]
Hereto, we distinguish between the two options for $\ell$. 
For $\ell \equiv n +1 \;(\bmod\; 3) $ this gives the conditions
\[
\Big(k+\frac{1}{2}\Big)^2- 9\kappa^2    \neq 0 \qquad\mbox{for } k\in\{-\lambda,\dotsc,\lambda-1\}\cap 3\mathbb{Z}\,,
\]
while $\ell \equiv n-1 \;(\bmod\; 3) $ also leads to
\[
\Big(k+\frac{1}{2}\Big)^2- 9\kappa^2  \neq 0 \qquad \mbox{for }k\in\{-\lambda,\dotsc,\lambda-1\}\cap 3\mathbb{Z}\,.
\]
This shows that $\mathcal{B}$ forms a basis for an irreducible $\mathcal{O}\mathrm{S}t_3$ representation if and only if $\kappa$ is not allowed to take on some specific values.

\subsubsection{Odd \texorpdfstring{$n$}{n}}
\label{sec:4.3.2}

Next, we consider the case where $n$ is an odd integer. As $\lambda=n/2$ is now a half-integer there exists an integer value $k_0 = \lambda -1/2= (n -1)/2 $ such that 
\[
O_0  v_{k_0}^+ = \frac12 v_{k_0}^+ ,\qquad O_0  v_{k_0+1}^- = \frac12 v_{k_0+1}^-  ,\qquad O_0  v_{k_0+1}^+ = -\frac12 v_{k_0+1}^+  ,\qquad O_0  v_{k_0}^- = -\frac12 v_{k_0}^-  \,, 
\] 
These specific eigenvalues have as a consequence that the acquired actions~\eqref{actOmp},\eqref{actOmm},\eqref{actOpm},\eqref{actOpp} do not lead to the full action of $O_-$ or $O_+$, as we would have to divide by zero. 
Using~\eqref{KpmOpm}, however, we find
\[
v_{k_0+1}^+ = K_- v_{k_0}^+ =  O_{-}\Big(O_0  - \frac12\Big) v_{k_0}^+  + \kappa O_{123}N_{-} v_{k_0}^+ =  3\kappa \Lambda \mathbf{1}_{3\mathbb{Z}}(\ell+ k_0- 1) v_{k_0}^-\,.
\]
Since $k_0< n$, the action $K_- v_{k_0}^+$ may not result in zero by the assumption of minimality on $n$, so we must have $\mathbf{1}_{3\mathbb{Z}}(\ell+ k_0- 1)=1$ or thus $\ell \equiv 2k_0 +1 \equiv n \;(\bmod\; 3) $. It follows that the vectors $v_{k_0+1}^+$ and $ v_{k_0}^-$, which have the same $O_0$ eigenvalue, are not linearly independent as now
$
v_{k_0+1}^+ =   3\kappa \Lambda v_{k_0}^-
$. 
In the same way, we find $
v_{k_0+1}^- =   3\kappa \Lambda v_{k_0}^+
$. By means of these results and the actions~\eqref{actKp} and
\eqref{actKm} of $K_{\pm}$ we obtain that the vector $v_k^-$ is proportional to $v_{n-k}^+$ for every $k\in\{0,1,\dotsc,n\}$. 
Indeed, by~\eqref{actKm} we have for instance
\[
v_{k_0+2}^+ =K_- v_{k_0+1}^+ = 3\kappa \Lambda \,K_-  v_{k_0}^- = -3\kappa \Lambda \big(\Lambda^2
+ 1  \big)
v_{k_0-1}^- \,.
\]
However, acting on $v_{k_0}^+$ with $[O_0,O_-]$, see relation~\eqref{O0Opm}, we find an equation which can never be satisfied unless $v_{k_0+1}^+=0$. Hence, we have no representations for odd $n$ in case (c).

\subsection{Unitary representations} 
\label{sec:unitary}

To find irreducible unitary representations we check which of the irreducible representations admit an invariant positive definite Hermitian form. Hereto, we introduce an 
antilinear antimultiplicative involution $X\mapsto X^{\dagger}$ compatible with the algebraic relations~\eqref{OijOki} of the algebra $\mathcal{O}\mathrm{S}t_3$. This involution has the properties 
$(aX+bY)^{\dagger} = \overline{a} X^{\dagger} + \overline{b} Y^{\dagger}$ and $(XY)^{\dagger} = Y^{\dagger} X^{\dagger}$ for $X,Y \in \mathcal{O}\mathrm{S}t_3$ and $a,b\in\mathbb{C}$, where $\overline{a}$ denotes complex conjugation. 

For real $\kappa$, the algebraic relations~\eqref{OijOki} are compatible with the star conditions 
\begin{equation*}
O_{12}^{\dagger} = - O_{12}\qquad O_{23}^{\dagger} = - O_{23}\qquad O_{31}^{\dagger} = - O_{31}\qquad O_{123}^{\dagger} = - O_{123}\,.
\end{equation*}
and
\begin{equation*}
G_{12}^{\dagger} =  G_{12}\qquad G_{23}^{\dagger} =  G_{23}\qquad G_{31}^{\dagger} =  G_{31}\qquad G_{123}^{\dagger} =  G_{321}\,.
\end{equation*}
\begin{rema}
	Note that the total angular momentum operators $O_{12},O_{23},O_{31}$ become self-adjoint when accompanied by the factor $1/\ii$ we have left out for notational convenience. 
\end{rema}
In terms of Definition~\ref{def1}, this leads to the relations~\eqref{O0Opm}--\eqref{OpOm} being compatible with the star conditions
\begin{equation}\label{star}
O_{123}^{\dagger} = - O_{123}\qquad O_0^{\dagger} = O_0, \qquad  O_{\pm}^{\dagger} = O_{\mp},\qquad K_{\pm}^{\dagger} = K_{\mp},\qquad N_{\pm}^{\dagger} = N_{\mp}\,.
\end{equation}
We show that if the value of $\kappa$ is suitably restricted, the representation $V$ is unitary under~\eqref{star}. 
Hereto, we introduce a sesquilinear form $\langle \cdot,\cdot \rangle \colon V \times V \to \mathbb{C}$ such that for $X\in \mathcal{O}\mathrm{S}t_3$ and $v,w\in V$
\[
\langle X v, w \rangle = \langle v, X^{\dagger} w \rangle\rlap{\,.}
\]
The condition 
$O_0^{\dagger} = O_0$ implies that vectors with different $O_0$ eigenvalues are orthogonal, so the previously determined bases are in fact orthogonal. Hence, we may define the form $\langle\cdot,\cdot\rangle$ by putting

\[
\langle v_k^{+}, v_{l}^{+} \rangle = h_k\,  \delta_{k,l}  \rlap{\,,}\qquad 
\langle v_k^{+}, v_{l}^{-} \rangle = 0 \rlap{\,,}
\]
where we can freely let $h_0=1$ or $\langle v_0^+, v_0^+ \rangle = 1 $. 
Note that 
\[
\langle v_k^{-}, v_{l}^{-} \rangle = \langle G_{12} v_k^{+},  G_{12} v_{l}^{+} \rangle  = \langle G_{12} G_{12} v_k^{+},  v_{l}^{+} \rangle = \langle v_k^{+},   v_{l}^{+} \rangle = h_k\,  \delta_{k,l}\,.
\]

In order to be an inner product we need $h_k>0$ for $k\geq 0$. 
Using the star condition $K_-^{\dagger}= K_+$ and using $K_+v_k^+=A(k)v_{k-1}^+$ with~\eqref{A}, we have for $k\geq 1$ 
\begin{equation}
\label{hk}
h_{k} = \langle v_{k}^+, v_{k}^+ \rangle = \langle K_- v_{k-1}^+,  v_{k}^+ \rangle  = \langle  v_{k-1}^+, K_+ v_{k}^+ \rangle = A(k)  \langle  v_{k-1}^+, v_{k-1}^+ \rangle = A(k) h_{k-1}  \rlap{\,.}
\end{equation}
In this way we arrive at the condition $A(k)>0$ for $1 \leq k\leq n$, which is obviously satisfied for the case (a) with the choice $\lambda=n+3\kappa+1/2$. 
This will constitute the only class of unitary representations without further restrictions on the non-negative parameter $\kappa$. 
For the other choice of case (a), $\lambda=n-3\kappa+1/2$, this only holds when $\kappa$ is restricted to $|\kappa| < 1/3$. 
For the case (b), we have two options for $\ell$, leading to different restrictions on the value of $\kappa$ in order for $A(k)>0$ to hold for $1 \leq k\leq n$. 
If $\ell \equiv 2n \;(\bmod\; 3) $, then $\kappa$ must satisfy $|\kappa| < 2/3$, 
while $\ell \equiv 2n+2 \;(\bmod\; 3) $ implies the condition $|\kappa| < 1/3$. 
For the case (c) with $n$ even we have $|\kappa| < 1/2$ if $\ell \equiv n \;(\bmod\; 3) $ and 
$|\kappa| < 1/6$
if $\ell \not\equiv n \;(\bmod\; 3) $.

Given an inner product we can introduce the orthonormal basis 
\[
w_{k}^{\pm} = \frac{v_{n-k}^{\pm}}{\lVert v_{n-k}^{\pm} \rVert} \qquad (k = 0, 1, \dots , n-1, n)
\]
where $\lVert v_{n-k}^{\pm} \rVert = \sqrt{ \langle  v_{n-k}^{\pm}, v_{n-k}^{\pm} \rangle }  = \sqrt{ h_{n-k} } $. We find using \eqref{hk}
\[
K_- w_{k}^{+} = K_- \frac{v_{n-k}^{+}}{\lVert v_{n-k}^{+} \rVert} 
= \frac{v_{n-k+1}^{+}}{ \sqrt{ h_{n-k} }}  
= \sqrt{ A(n-k+1) } \,w_{k-1}^{+}
\]
and by~\eqref{actK}
\[
K_+ w_{k}^{+}= K_+  \frac{v_{n-k}^{+}}{\lVert v_{n-k}^{+} \rVert}   = A(n-k) \frac{v_{n-k-1}^{+}}{ \sqrt{ h_{n-k} }}  = \sqrt{ A(n-k) } \,w_{k+1}^{+} \rlap{\,,}
\]
while similarly
\[
K_- w_{k}^{-} = K_- \frac{v_{n-k}^{-}}{\lVert v_{n-k}^{-} \rVert} 
= A(n-k)\frac{v_{n-k-1}^{-}}{ \sqrt{ h_{n-k} }}  
= \sqrt{ A(n-k) } \,w_{k+1}^{-}
\]
and
\[
K_+ w_{k}^{-}= K_+  \frac{v_{n-k}^{-}}{\lVert v_{n-k}^{-} \rVert}   =  \frac{v_{n-k+1}^{-}}{ \sqrt{ h_{n-k} }}  = \sqrt{ A(n-k+1) } \,w_{k-1}^{-} \rlap{\,.}
\]

Returning to the case (a), the right-hand side follows from
\begin{equation*}
A(k) = k(2n+6\kappa+2-k)  
\big((n+3\kappa-k+1)^2- 9\kappa^2 \mathbf{1}_{3\mathbb{Z}}(2n+k+2)  \big)\rlap{\,.}
\end{equation*}
We summarize all actions for the case (a) in the following proposition. 

\begin{prop}\label{prop1}
	For a given positive parameter $\kappa$ and a choice of sign $\varepsilon=\pm 1$, we have an irreducible representation of $\mathcal{O}\mathrm{S}t_3$ of dimension $2n+2$ for every non-negative integer $n$. This representation is unitary, corresponding to the star conditions~\eqref{star}.  The actions of the $\mathcal{O}\mathrm{S}t_3$ operators on a set of basis vectors 
	$w_0^+$, $w_1^+$, $\ldots$, $w_n^+$ and $w_0^-$, $w_1^-$, $\ldots$, $w_n^-$ are given by:
	\begin{align}
	& O_0 w_k^{\pm} =  \pm\Big(k  + \frac12+3\kappa\Big) \; w_k^{\pm}\label{act-O0}\\
	& O_{123} w_k^{\pm} =  \varepsilon\, \ii\, (n+1 +3\kappa) \; w_k^{\pm}\label{act-O123}\\
	& K_{+} w_k^{+} =  
	\begin{cases}
	\sqrt{(k+1) (n-k)(n+k+2+6\kappa) 
		(k+1+6\kappa)}  w_{k+ 1}^{+}
	& \text{if }k\equiv 2 \;(\bmod\; 3)\\
	(k+1+3\kappa) \sqrt{(n-k)(n+k+2+6\kappa)}w_{k+ 1}^{+}
	& \text{if }k\not\equiv 2 \;(\bmod\; 3)
	\end{cases} \label{act-K++}\\
	& K_{+} w_k^{-}= 
	\begin{cases}
	\sqrt{ k(n-k+1)(n+k+1+6\kappa)(k+6\kappa)} w_{k- 1}^{+}
	& \qquad\text{if }k\equiv 0 \;(\bmod\; 3)\\
	(k+3\kappa)\sqrt{(n-k+1)(n+k+1+6\kappa) }w_{k- 1}^{+}
	& \qquad\text{if }k\not\equiv 0 \;(\bmod\; 3)
	\end{cases} \label{act-K+-}\\		
	& K_{-} w_k^{+}=
	\begin{cases}
	\sqrt{ k(n-k+1)(n+k+1+6\kappa)(k+6\kappa)} w_{k- 1}^{+}
	& \qquad\text{if }k\equiv 0 \;(\bmod\; 3)\\
	(k+3\kappa)\sqrt{(n-k+1)(n+k+1+6\kappa) } w_{k- 1}^{+}
	& \qquad\text{if }k\not\equiv 0 \;(\bmod\; 3)
	\end{cases} \label{act-K-+}\\
	& K_{-} w_k^{-}=
	\begin{cases}
	\sqrt{ (k+1)(n-k)(n+k+2+6\kappa)(k+1+6\kappa)}  w_{k+ 1}^{-}
	& \text{if }k\equiv 2 \;(\bmod\; 3)\\
	(k+1+3\kappa)\sqrt{(n-k)(n+k+2+6\kappa) }  w_{k+ 1}^{-}
	& \text{if }k\not\equiv 2 \;(\bmod\; 3)
	\end{cases} \label{act-K--}
	\end{align} 	 
	while for $O_+$ and $O_-$ we have the following actions. If $k\equiv 0 \;(\bmod\; 3)$ then 
	\begin{align}
	O_{+} w_k^{+}&=
	\sqrt{ (n-k)(n+k+2+6\kappa)} \; w_{k+ 1}^{+}
	\\
	O_{+} w_k^{-}&=
	-\frac{\sqrt{ k(n-k+1)(n+k+1+6\kappa)(k+6\kappa)}}{k+3\kappa} w_{k- 1}^{-}
	+\varepsilon\, \ii\,\frac{ 3\kappa(n+1+3\kappa)}{k+3\kappa}w_{k}^{+}  
	\\
	O_{-} w_k^{+}&=
	\frac{\sqrt{ k(n-k+1)(n+k+1+6\kappa)(k+6\kappa)}}{k+3\kappa} w_{k- 1}^{+}
	-\varepsilon\, \ii\,\frac{ 3\kappa(n+1+3\kappa)}{k+3\kappa}w_{k}^{-}  
	\\
	O_{-} w_k^{-}&=
	-\sqrt{ (n-k)(n+k+2+6\kappa)} \; w_{k+ 1}^{-}\rlap{\,;}
	\end{align}	
	if $k\equiv 1 \;(\bmod\; 3)$ then
	\begin{align}
	O_{+} w_k^{+}&=
	\sqrt{ (n-k)(n+k+2+6\kappa)} \; w_{k+ 1}^{+}\\
	O_{+} w_k^{-}&=	-\sqrt{(n-k+1)(n+k+1+6\kappa) } \; w_{k- 1}^{-}\\
	O_{-} w_k^{+}&=	\sqrt{(n-k+1)(n+k+1+6\kappa) } \; w_{k- 1}^{+}\\
	O_{-} w_k^{-}&=
	-\sqrt{ (n-k)(n+k+2+6\kappa)} \; w_{k+ 1}^{-}\rlap{\,;}
	\end{align}	
	if $k\equiv 2 \;(\bmod\; 3)$ then 
	\begin{align}
	O_{+} w_k^{+}&=
	\frac{\sqrt{ (k+1)(n-k)(n+k+2+6\kappa)(k+1+6\kappa)}}{k+1+3\kappa} w_{k+ 1}^{+}
	-\varepsilon\, \ii\,\frac{ 3\kappa(n+1+3\kappa)}{k+1+3\kappa}w_{k}^{-}
	\\
	O_{+} w_k^{-}&=	-\sqrt{(n-k+1)(n+k+1+6\kappa) } \; w_{k- 1}^{-}
	\\
	O_{-} w_k^{+}&=	\sqrt{(n-k+1)(n+k+1+6\kappa) } \; w_{k- 1}^{+}
	\\
	O_{-} w_k^{-}&=
	-\frac{\sqrt{ (k+1)(n-k)(n+k+2+6\kappa)(k+1+6\kappa)}}{k+1+3\kappa} w_{k+ 1}^{-}
	+\varepsilon\, \ii\,\frac{ 3\kappa(n+1+3\kappa)}{k+1+3\kappa}w_{k}^{+}\rlap{\,.}
	\end{align}	
	For the realization of $\mathrm{S}_3$ within $\mathcal{O}\mathrm{S}t_3$ we have the actions 
	\begin{align}
	G_{12} w_k^{\pm} 	&=  \; w_k^{\mp}\label{act-G12} &
	G_{23} w_k^{\pm} 	&=\omega^{\pm(1-k)}  \; w_k^{\mp} &
	G_{31} w_k^{\pm} 	&= \omega^{\pm(  k-1 )} \; w_k^{\mp}	\\
	& &	 G_{123} w_k^{\pm} & = \omega^{\pm(1-k)} \; w_k^{\pm}\label{act-G123}&	
	G_{321} w_k^{\pm} &= \omega^{\pm(  k-1 )} \; w_k^{\pm}\rlap{\,.}
	\end{align}	
\end{prop}

We thought it to be instructive to include a diagram depicting the basis vectors and actions of Proposition~\ref{prop1} according to their eigenvalues for $O_0$ and $G_{123}$, given in Figure~\ref{FigRep}.  Denoting the action~\eqref{act-O0} by $O_0w_k^{\pm} = \pm \lambda_k w_k^{\pm}$, the distance between $\lambda_0$ and $-\lambda_0$ on the horizontal axis is $6\kappa+1$, and thus depends on the value of the parameter $\kappa$.

\section{Explicit realizations}\label{sec:5}

The irreducible unitary representations of case (a) as classified above have an explicit realization in the framework of Dunkl operators~\eqref{Dunkl}. Indeed, in the original construction of the algebra the symmetries $O_{12},O_{23},O_{31}$ consist of Dunkl angular momentum operators with added reflection terms, see~\eqref{Oij}.
When a symmetry (anti)commuting with the Dirac-Dunkl operator $\uD$ acts on an element in the kernel of $\uD$, the result is again in this kernel. Furthermore, as the symmetries $O_{12},O_{23},O_{31}$ are grade-preserving, it is no surprise that homogeneous polynomials of fixed degree in $\ker \uD$ will form the desired representation spaces. 

We will first introduce some notations and definitions. Let $\mathcal{P}_n(\mathbb{R}^N)$ denote the space of homogeneous polynomials of degree $n$ in $N$ variables. The Dunkl monogenics of degree $n$ are homogeneous spinor-valued polynomials of degree $n$ in the kernel of the Dirac-Dunkl operator, which we will denote by $\mathcal{M}_n(\mathbb{R}^N,\mathbb{S}) = \ker\uD \cap (\mathcal{P}_n(\mathbb{R}^N)\otimes \mathbb{S})$. Here $\mathbb{S}$ is a spinor representation of the Clifford algebra. For the three-dimensional Clifford algebra realized by the Pauli matrices, 
a two-dimensional Dirac spinor representation is simply $\mathbb{S}\cong \mathbb{C}^2$, with basis spinors $\chi^+ = (1,0)^T$ and $\chi^- = (0,1)^T$. 

The Dunkl monogenics form eigenspaces of the angular Dirac-Dunkl operator $\Gamma$. 
Indeed, for $\psi_n \in \mathcal{M}_n(\mathbb{R}^3,\mathbb{S})$ we have using $\uD\psi_n = 0$  and~\eqref{Euler}
\begin{align*}
(\Gamma+1) \psi_n &= \frac12(  [\uD, \ux ] -1)\psi_n 
\ = \frac12(  \uD \,\ux   -1)\psi_n \\*
&= \frac12(  \{\uD, \ux \}  -1)\psi_n = \frac12( 2\mE +3 + 6\kappa  -1)\psi_n\rlap{\,.}
\end{align*}
As $\mE \psi_n = n \psi_n$, this gives the following eigenvalues
\begin{equation}\label{eigGamma}
(\Gamma+1) \psi_n  =( n+1 +3\kappa) \psi_n \rlap{\,.}
\end{equation}
Keeping in mind the relation $O_{123} = (\Gamma+1)e_1e_2e_3$, and comparing these eigenvalues with the action~\eqref{act-O123}, this confirms our expectation regarding realizations of the obtained representations. 

Appending a factor $(1+e_0)$ to a Dunkl monogenic $\psi_n\in \mathcal{M}_n(\mathbb{R}^3,\mathbb{S})$, we obtain a (rather trivial) eigen\-function of the Dunkl Dirac Hamiltonian~\eqref{DiracDunklHam}. Indeed, using the anticommutaton relations of the Clifford algebra, we find
\[
H_{\kappa} (1+e_0)\psi_n = \uD(1+e_0)\psi_n + me_0(1+e_0)\psi_n =  (1-e_0)\uD\psi_n + m(e_0+1)\psi_n = m(e_0+1)\psi_n\rlap{\,.}
\]
Note that $(1+e_0)\psi_n$ is no longer an eigenfunction of $O_{123}$, as the latter does not commute, but anticommutes with $H_{\kappa}$.

We will set out to construct a basis for the space of Dunkl monogenics. 
Hereto, it is useful to emulate a setting similar to that of Definition~\ref{def1} and
Proposition~\ref{Result1} by means of a coordinate change:
\begin{equation}\label{coordchange}
\begin{pmatrix}  u \\ v\\ w \end{pmatrix} = 
\begin{pmatrix} 
\frac{1}{\sqrt{2}} & \frac{-1}{\sqrt{2}} &   0 \\ 
\frac{1}{\sqrt{6}} & \frac{1}{\sqrt{6}} & \frac{-2}{\sqrt{6}}\\ 
\frac{1}{\sqrt{3}} & \frac{1}{\sqrt{3}} & \frac{1}{\sqrt{3}} 
\end{pmatrix} 
\begin{pmatrix} x_1 \\ x_2\\ x_3 \end{pmatrix} ,
\qquad 
\begin{pmatrix} x_1 \\ x_2\\ x_3 \end{pmatrix} 
=   \begin{pmatrix} 
\frac{1}{\sqrt{2}} & \frac{1}{\sqrt{6}} & \frac{1}{\sqrt{3}} \\ 
\frac{-1}{\sqrt{2}} & \frac{1}{\sqrt{6}} & \frac{1}{\sqrt{3}}\\ 
0 & \frac{-2}{\sqrt{6}} & \frac{1}{\sqrt{3}} 
\end{pmatrix} 
\begin{pmatrix}  u \\ v\\ w \end{pmatrix} \,.
\end{equation}
The action of $g_{12}$ on functions of $(u,v,w)$ becomes very simple, flipping only the sign of $u$,  $g_{12}f(u,v,w) = f(-u,v,w)$, while the other transpositions $g_{23}$ and $g_{31}$ act as follows
\[
g_{23}f(u,v,w) = f\Big(\frac12 u + \frac{\sqrt{3}}{2} v,\frac{\sqrt{3}}{2} u - \frac12 v,w\Big),
\quad 
g_{31}f(u,v,w) = f\Big(\frac12 u - \frac{\sqrt{3}}{2} v, -\frac{\sqrt{3}}{2}u -  \frac12v,w\Big)\rlap{\,.}
\]
For the Dunkl operators associated to this new coordinate basis we find the following explicit expressions: we have $\cD_w = \partial_w $, while
\[
\cD_u = \partial_u + \kappa \left( \frac{1-g_{12}}{u} + \frac{1-g_{23}}{u-\sqrt{3}v} + \frac{1-g_{31}}{u+\sqrt{3}v}    \right),
\quad
\cD_v= \partial_v + \kappa \left(  \sqrt{3}\frac{1-g_{23}}{-u+\sqrt{3}v} + \sqrt{3}\frac{1-g_{31}}{u+\sqrt{3}v}    \right)\rlap{\,.}
\]
The commutation relations of $\cD_u,\cD_v,\cD_w$ and $u,v,w$ are given in Table~\ref{tab:2}.
\begin{table}[!hb]
	\caption{Commutation relations $\cD_u,\cD_v,\cD_w$ and $u,v,w$.}
	\label{tab:2}       
	\[
	\begin{array}{c|cccc}
	[\downarrow,\rightarrow] & u & v & \qquad w\qquad   \\ \hline 
	\cD_u & 1+\kappa(2g_{12}+\frac12g_{23}+\frac12g_{31})& -\kappa\frac{\sqrt{3}}{2}(g_{23}-g_{31})  & 0  \\
	\cD_v & -\kappa\frac{\sqrt{3}}{2}(g_{23}-g_{31}) & 1+\kappa(\frac32g_{23}+\frac32g_{31}) & 0 \\
	\cD_w  & 0 & 0 & 1 
	\end{array}
	\]
\end{table}
We see that in the coordinate frame of $u,v,w$ the action of the reflection group is restricted to the $(u,v)$-plane.

As $u,v,w$ form again an orthonormal basis of $\mathbb{R}^3$, the Laplace-Dunkl operator~\eqref{Laplace} can also be written as
\[
\Delta = \cD_u^2 +  \cD_v^2 +  \cD_w^2\rlap{\,.}
\]
By applying the same coordinate change~\eqref{coordchange} to the Clifford generators $e_1,e_2,e_3$, that is
\[
e_u =  \frac{1}{\sqrt2} (e_1-e_2),\quad e_v= \frac{1}{\sqrt6} (e_1 + e_2-2e_3),\quad e_w = \frac{1}{\sqrt3} (e_1+e_2+e_3)\rlap{\,,}
\]
the Dirac-Dunkl operator can now be written as
\[
\uD = e_u  \cD_u +  e_v \cD_v + e_w  \cD_w\rlap{\,.}
\]
Similarly, in these new coordinates the vector variable becomes $
\ux = u e_u +  v e_v +  w e_w
$ which squares to $\ux^2 =u^2 +v^2 +w^2$ and the Euler operator is given by
$\mathbb{E} = u\partial_u +v\partial_v +w\partial_w$.
The triple $e_u,e_v,e_w$ forms another basis of the Euclidean Clifford algebra since one readily verifies by means of the anticommutation relations of $e_1,e_2,e_3$ that also
\[
e_u^2 = e_v^2=e_w^2 =1 ,\qquad \{e_u,e_v\}  = \{e_v,e_w\} = \{e_w,e_u\} = 0\rlap{\,.}
\]
For practical purposes, we will realize $e_u,e_v,e_w$ by the Pauli matrices 
\begin{equation}\label{pauli}
e_u = \begin{pmatrix}0& 1 \\ 1 & 0 \end{pmatrix},\qquad e_v = \begin{pmatrix}0& -\ii \\ \ii & 0 \end{pmatrix},\qquad e_w=\begin{pmatrix}1& 0 \\ 0 & -1 \end{pmatrix}\rlap{\,.}
\end{equation}

The generators of the realization of $\mathrm{S}_3$ within $\mathcal{O}\mathrm{S}t_3$ in this framework become
\begin{equation*}
G_{12} = g_{12} e_u\rlap{\,,}\qquad G_{23} = g_{23}\frac12( - e_u+\sqrt{3}e_v)\rlap{\,,}\qquad G_{31} = g_{31}\frac12(- e_u -\sqrt{3}e_v)\rlap{\,,}
\end{equation*}
all of which anticommute with $\uD$. In terms of the Pauli matrices, we have
\begin{equation}\label{Gijuvw}
G_{12} = g_{12} \begin{pmatrix}0& 1 \\ 1 & 0 \end{pmatrix} \rlap{\,,}\qquad 
G_{23} = g_{23}\begin{pmatrix}0&  \omega^2 \\ \omega & 0 \end{pmatrix}\rlap{\,,}\qquad 
G_{31} = g_{31}\begin{pmatrix}0&  \omega \\  \omega^2 & 0 \end{pmatrix}\rlap{\,.}
\end{equation}
Similar to~\eqref{Oij}, in the $u,v,w$ coordinates we obtain the following symmetries commuting with $\uD$:  
\begin{align}
O_{uv} & = u \cD_v - v \cD_u + \frac12 e_ue_v +\kappa e_u e_v (g_{12} + g_{23}+g_{31})\rlap{\,,}\label{Ouv}\\
O_{vw} & = v \cD_w - w \cD_v + \frac12 e_ve_w +\kappa \frac{3}{4} e_v e_w (g_{23}+g_{31}) + \kappa \frac{\sqrt3}{4}e_w e_u(g_{23}-g_{31})\rlap{\,,}\label{Ovw}\\
O_{wu} & = w \cD_u - u \cD_w + \frac12 e_we_u +\kappa e_w e_u g_{12} + \kappa \frac14 e_u e_v(g_{23}+g_{31}) + \kappa \frac{\sqrt3}{4}e_v e_w(g_{23}-g_{31})\rlap{\,.}\label{Owu}
\end{align}
By direct verification after applying the coordinate change~\eqref{coordchange}, the operators of Definition~\ref{def1} now turn out to be
\[
O_0 = -\ii O_{uv},\qquad O_+ = \ii O_{wu} + O_{vw},\qquad  O_- = \ii O_{wu} - O_{vw}\rlap{\,,}
\]
and $N_{\pm}$ follows from the new expressions for the transpositions~\eqref{Gijuvw}, while 
\[
O_{123} =  -\frac12e_ue_ve_w -  \kappa (g_{12}+g_{23}+g_{31})e_ue_ve_w +  O_{uv} e_w  +   O_{vw} e_u + O_{wu}  e_v \rlap{\,.}
\]
The angular Dirac-Dunkl operator $\Gamma$ is again related to $O_{123}$, we have $O_{123}=(\Gamma+1)e_ue_ve_w$.

\subsection{A basis for the space of Dunkl monogenics}
\label{sec:5.1}

Next, we construct the vectors upon which these operators act. As already alluded to, the representation space will consist of Dunkl monogenics, homogeneous polynomials in the kernel of $\uD$. Except for the lowest degree or dimension, finding explicit expressions for a basis of the space of Dunkl monogenics is far from trivial. For an abelian reflection group, as in ref.~\cite{DeBie&Genest&Vinet-2016}, one can single out coordinates and, starting from polynomials on $\mathbb{R}$, gradually work up in dimension by means of Cauchy-Kowalevsky extension maps. 
For a non-abelian reflection group $G$, however, one is not able to single out coordinates at will, as the orbits of the action, or the conjugacy classes, of $G$ are not singleton sets. 
The advantage of the coordinate change~\eqref{coordchange} is that the coordinate $w$ does become invariant under all reflections. This means that for the coordinate $w$ we do in fact have a Cauchy-Kowalevsky extension map (see Proposition~\ref{propCK}) which allows us to move from two-dimensional space to three dimensions. 
On $\mathbb{R}^2$, Dunkl monogenics follow from the expressions for the Dunkl harmonics which were determined already in~\cite{1989_Dunkl_TransAmerMathSoc_311_167}.

When working in $\mathbb{R}^2$ spanned by the coordinates $u$ and $v$, it is useful to have a separate notation for the two-dimensional analogues of the Dirac-Dunkl operator, vector variable and Laplace-Dunkl operator:
\begin{equation}\label{tilde}
\tilde\uD =  e_u \cD_u +  e_v\cD_v \,,
\qquad
\tilde\ux = e_u u  +   e_v v\,,
\qquad
\tilde\Delta = \cD_u^2+\cD_v^2 =\tilde\uD^2,
\qquad
\tilde{\ux}^2 = u^2 +v^2\,.
\end{equation}
They satisfy the (anti)commutation relations, readily verified by means of the relations in Table~\ref{tab:2}, 
\begin{equation}\label{tilderel}
[\tilde\uD,\tilde\ux^2] = 2\tilde\ux,\qquad \{\tilde\uD,\tilde\ux\} = 2(\tilde\mE +1 + 3\kappa),\qquad \tilde\mE = u\partial_u +v\partial_v\rlap{\,,}
\end{equation}
where the Euler operator $\tilde{\mE}$ when acting on a polynomial measures the degree in $u$ and $v$. 

Finally, for the following proposition, the hypergeometric series~\cite{Bailey,Slater} is defined as
\begin{equation}
{}_2F_1 \left( \atop{a,b}{c} ; z \right)=\sum_{k=0}^\infty \frac{(a)_k(b)_k}{(c)_k}\frac{z^k}{k!},
\label{defF}
\end{equation}
where we use the common notation for Pochhammer symbols~\cite{Bailey,Slater}: $(a)_0=1$ and 
$(a)_k=a(a+1)\cdots(a+k-1)$ for $k=1,2,\dotsc$. 

\begin{prop}
	For a non-negative integer $k$, the polynomials $\phi_k^+$ and $\phi^-_k$ defined as
	\begin{equation}\label{phipm}
	\phi_k^{\pm}(u,v) = (u\pm iv)^k \frac{(\kappa+1)_n}{n!}{\;}_2F_1 \left( \atop{-n,\kappa} 
	{-n-\kappa} ; \frac{(-u\pm iv)^3}{(u\pm iv)^3} \right), \qquad n=\lfloor k/3\rfloor
	\end{equation}
	form a basis for the space of Dunkl harmonics $\mathcal{H}_k(\mathbb{R}^2)= \ker\tilde\Delta \cap \mathcal{P}_k(\mathbb{R}^2)$.
\end{prop}
\begin{proof}On two dimensional space $\mathbb{R}^2$ the Laplace-Dunkl operator $\tilde\Delta$ can be factorized as
	\[ 
	\tilde\Delta = \cD_u^2+\cD_v^2 =(\cD_u+ \ii \cD_v)(\cD_u- \ii \cD_v)   \rlap{\,.}
	\]
	For reflection groups on $\mathbb{R}^2$, the analogues of harmonic polynomials for the Dunkl Laplacian 
	were determined explicitly already in~\cite{1989_Dunkl_TransAmerMathSoc_311_167}.
	The expression~\eqref{phipm} is the hypergeometric form of polynomials satisfying (see~\cite{1989_Dunkl_TransAmerMathSoc_311_167})
	\[ 
	(\cD_u+ \ii \cD_v) \phi_k^{+}(u,v) = 0, \qquad  (\cD_u- \ii \cD_v) \phi_k^{-}(u,v) = 0\rlap{\,,}
	\]
	and hence $  \tilde\Delta\phi_k^{\pm}(u,v) = 0 $. For $k\geq 1$ the dimension of $\mathcal{H}_k(\mathbb{R}^2)$ is 2 so $\phi_k^+$ and $\phi_k^-$ form a basis, while the dimension of $\mathcal{H}_0(\mathbb{R}^2)$ is 1 in accordance with $\phi_0^+ = 1 = \phi_0^-$. 
\end{proof}

Note that the polynomial $\phi_k^{-}$ is simply the complex conjugate of $\phi_k^{+}$. These polynomials can also be written in terms of the Jacobi polynomials~\cite{Koekoek}, which are defined in terms of the hypergeometric series as 
\begin{equation}
\label{jacobi}
P_n^{\alpha,\beta}(x) = \frac{(\alpha+1)_n}{n!}{\;}_2F_1 \left( \atop{-n,n+\alpha+\beta+1} 
{\alpha+1} ; \frac{1-x}{2} \right)\rlap{\,.}
\end{equation}
By means of the identity
\begin{equation}\label{jacobident}
(x+y)^n P_n^{(\alpha,\beta)}\bigg(\frac{x-y}{x+y}\bigg) = \frac{(\alpha+1)_n}{n!} x^n {\;}_2F_1 \left( \atop{-n,-n-\beta} 
{\alpha+1} ; -\frac{y}{x} \right)\rlap{\,,}
\end{equation}
we can write~\eqref{phipm}, denoting $z=u+iv$ and $\overline{z} = u-iv$, as
\begin{equation}\label{jacob}
\phi_k^{+}(u,v) =  (-1)^n z^{k-3n} (z^3+\overline{z}^3)^n P_n^{(-n-\kappa-1,-n-\kappa)} \left( \frac{z^3-\overline{z}^3}{z^3+\overline{z}^3} \right),\qquad n=\lfloor k/3\rfloor\rlap{\,.}
\end{equation}

We use the previous result to obtain spinor-valued polynomials in the kernel of the two-dimensional Dirac-Dunkl operator
$
\tilde\uD =  e_u \cD_u +  e_v\cD_v 
$. 
Recall that for the three-dimensional Clifford algebra realized by the Pauli matrices, 
a two-dimensional Dirac spinor representation is $\mathbb{S}\cong \mathbb{C}^2$, with basis spinors $\chi^+ = (1,0)^T$ and $\chi^- = (0,1)^T$. 

\begin{prop}\label{propMk2} 
	For a non-negative integer $k$, the polynomials
	\begin{equation}\label{varphipm}
	\varphi_k^{+}(u,v) =  	\phi_k^{+}(u,v)\chi^+ \quad\mbox{and}\quad \varphi_k^{-}(u,v) = \phi_k^{-}(u,v)\chi^- 
	\end{equation}
	form a basis for the space $\mathcal{M}_k(\mathbb{R}^2,\mathbb{C}^2)$.
\end{prop}
\begin{proof}
	Acting with
	$
	\tilde\uD = e_u\cD_u + e_v \cD_v
	$	
	on $\varphi_k^{+}$ we find using the Pauli matrices~\eqref{pauli}
	\[
	\tilde\uD\varphi_k^{+}(u,v)= \begin{pmatrix}0& 1 \\ 1 & 0 \end{pmatrix} \cD_u 	\phi_k^{+}(u,v)\begin{pmatrix}1\\0 \end{pmatrix}+  \begin{pmatrix}0& -\ii \\ \ii & 0 \end{pmatrix}\cD_v 	\phi_k^{+}(u,v)\begin{pmatrix}1\\0 \end{pmatrix}=  (\cD_u +\ii \cD_v )	\phi_k^{+}(u,v)\begin{pmatrix}0\\1 \end{pmatrix}
	\]
	which vanishes by definition of $\phi_k^{+}$. In the same way we find
	$
	\tilde\uD  \varphi_k^{-}(u,v) =0$. As the dimension of $\mathcal{M}_k(\mathbb{R}^2,\mathbb{C}^2)$ is 2, $\varphi_k^+$ and $\varphi_k^-$ form a basis.
\end{proof}

For non-negative $\kappa$, there exists a Fischer decomposition for Dunkl monogenics in the sense of the following direct sum decomposition
\[
\cP_n(\mathbb{R}^2)\otimes \mathbb{C}^2 = \bigoplus_{k=0}^n \tilde \ux^{n-k} \mathcal{M}_n(\mathbb{R}^2,\mathbb{C}^2)\rlap{\,.}
\]
Every spinor-valued polynomial on $\mathbb{R}^2$ can thus be written in terms of Dunkl monogenics on $\mathbb{R}^2$, for which a basis is given in Proposition~\ref{propMk2}. The next step consists of moving from $\mathbb{R}^2$ to Dunkl monogenics on $\mathbb{R}^3$ by means of a Cauchy-Kowalevski isomorphism. 

\begin{prop}\label{propCK}
	For a non-negative integer $n$, a basis for the space $\mathcal{M}_n(\mathbb{R}^3,\mathbb{C}^2)$ is given by the $2n+2$ polynomials
	\begin{equation}\label{psink}
	\psi_{n,k}^{\pm}(u,v,w) = \mathbf{CK}_w \big[ \tilde{\ux}^{n-k} \varphi_k^{\pm}(u,v) \big],\qquad k\in\{0,1,\dotsc,n\}
	\end{equation}	
	where the Cauchy-Kowalevski isomorphism is given by
	\begin{equation}\label{CK}
	\mathbf{CK}_w  \colon \cP_n(\mathbb{R}^2)\otimes \mathbb{C}^2  \to  \mathcal{M}_n(\mathbb{R}^3,\mathbb{C}^2)  \colon
	p_n(u,v)  \mapsto  \exp(-we_w\tilde{\uD})p_n(u,v)\rlap{\,.}
	\end{equation}	
	Note that as $p_n(u,v)$  is a polynomial of degree $n$, this reduces to the finite sum
	\[
	\mathbf{CK}_w\big[p_n(u,v) \big] = \exp(-we_w\tilde{\uD})p_n(u,v) = \sum_{a=0}^{n} \frac{(-1)^{a}}{a!}w^a(e_w\tilde{\uD})^a p_n(u,v)\rlap{\,.}
	\]
\end{prop}
\begin{proof}We show that the Cauchy-Kowalevski extension $\mathbf{CK}_w$ maps $\cP_n(\mathbb{R}^2)\otimes \mathbb{C}^2 $ into $\mathcal{M}_n(\mathbb{R}^3,\mathbb{C}^2)$. 
	Let $p_n(u,v)\in \cP_n(\mathbb{R}^2)\otimes \mathbb{C}^2$. Using $\cD_w = \partial_w$ and the commutation relations in Table~\ref{tab:2} we obtain
	\begin{align*}
	\uD  \,\mathbf{CK}_w\big[p_n(u,v) \big] & = (\tilde{\uD} + e_w\partial_w) \sum_{a=0}^{n} \frac{(-1)^{a}}{a!}w^a(e_w\tilde{\uD})^a p_n(u,v)\\
	& =  \sum_{a=0}^{n-1} \frac{(-1)^{a}}{a!}w^ae_w(e_w\tilde{\uD})^{a+1} p_n(u,v)
	+
	\sum_{a=1}^{n} \frac{(-1)^{a}}{(a-1)!}w^{a-1}e_w(e_w\tilde{\uD})^a p_n(u,v)
	\end{align*}
	which clearly vanishes. Hence, as the map~\eqref{CK} preserves the degree of a polynomial we have $\mathbf{CK}_w\big[p_n(u,v) \big] \in \mathcal{M}_n(\mathbb{R}^3,\mathbb{C}^2)$.
	
	The inverse of the isomorphism $\mathbf{CK}_w$ is given by the map which evaluates a function in $w=0$. As the degree of a polynomial in $\mathcal{M}_n(\mathbb{R}^3,\mathbb{C}^2)$ is fixed, this inverse is clearly injective.
\end{proof}
Note that $\psi_{n,k}^{\pm}$ given by~\eqref{psink} can also be written in terms of Jacobi polynomials~\eqref{jacobi} by working out the explicit action of the map~\eqref{CK}. To achieve this, we first state a result, which follows from the commutation relations~\eqref{tilderel}. For $M_k \in\mathcal{M}_k(\mathbb{R}^2,\mathbb{C}^2)$ and non-negative integers $a,b$, 
\begin{equation}\label{ident}
\tilde\uD^{a} \tilde\ux^{b} M_k = d_{a,b}^k\, \tilde\ux^{b-a} M_k\rlap{\,,}
\end{equation}
where $d_{a,b}^k=0$ for $a >b$, and otherwise distinguishing between even and odd $a,b$ one has
\begin{align*}
d_{2\alpha,2\beta}^k &=2^{2\alpha}(-\beta)_\alpha (-\beta-k-3\kappa)_\alpha ,& d_{2\alpha,2\beta+1}^k &= 2^{2\alpha}(-\beta)_\alpha (-\beta-k-1-3\kappa)_\alpha ,\\ 
d_{2\alpha+1,2\beta}^k &= -2^{2\alpha+1} (-\beta)_{\alpha+1} (-\beta-k-3\kappa)_\alpha ,& d_{2\alpha+1,2\beta+1}^k & = - 2^{2\alpha+1}(-\beta)_\alpha (-\beta-k-1-3\kappa)_{\alpha+1}\rlap{\,.}
\end{align*}
Using now in turn the identity~\eqref{ident}, $\tilde \uD e_w = - e_w \tilde \uD$ 
, $(2\alpha)! = 2^{2\alpha} (1)_\alpha (1/2)_\alpha$ and the identity~\eqref{jacobident},  
we obtain
\begin{equation}\label{psiexp}
\psi_{n,k}^{\pm}(u,v,w)  = \Psi_{n-k}(\tilde{\ux},w) \varphi_k^{\pm}(u,v) 
\end{equation}
with $\varphi_k^{\pm}$ given by~\eqref{varphipm} (see also~\eqref{jacob}), and
\begin{align}\label{Psiexp}
& \Psi_{n-k}(\tilde{\ux},w) =    \frac{\beta !}{(\frac12)_\beta} (u^2+v^2+w^2)^{\beta} \\
& \times 
\begin{cases}
P_{\beta}^{(-\frac12,k+3\kappa)}\Big(\frac{u^2+v^2-w^2}{u^2+v^2+w^2}\Big) -\frac{e_ww\tilde{\ux}}{u^2+v^2+w^2} P_{\beta-1}^{(\frac12,k+1+3\kappa)}\Big(\frac{u^2+v^2-w^2}{u^2+v^2+w^2}\Big)  & \mbox{ if } n-k = 2\beta  \rlap{\,,} \\
\tilde{\ux} \, P_{\beta}^{(-\frac12,k+1+3\kappa)}\Big(\frac{u^2+v^2-w^2}{u^2+v^2+w^2}\Big) -e_ww\frac{\beta+k+1+3\kappa}{\beta+\frac12} P_{\beta}^{(\frac12,k+3\kappa)}\Big(\frac{u^2+v^2-w^2}{u^2+v^2+w^2}\Big) & \mbox{ if }  n-k = 2\beta+1 \rlap{\,.}
\end{cases}\notag
\end{align}

\subsection{Representations}
\label{sec:5.2}

Given a non-negative integer $n$, we show that the basis vectors $\psi_{n,k}^{\pm}$ for $k\in\{0,1,\dotsc,n\}$ transform irreducibly under the action of the algebra $\mathcal{O}\mathrm{S}t_3$. As the elements of $\mathcal{O}\mathrm{S}t_3$ (anti)commute with the Dirac-Dunkl operator, the kernel of the Dirac-Dunkl operator is invariant under the action of $\mathcal{O}\mathrm{S}t_3$. Furthermore, the elements of $\mathcal{O}\mathrm{S}t_3$ are grade-preserving so the space $\mathcal{M}_n(\mathbb{R}^3,\mathbb{C}^2)$ is invariant under the action of $\mathcal{O}\mathrm{S}t_3$. 

The spinor $\psi_{n,k}^{\pm}$ corresponds, up to rescaling, precisely to the basis vector
$w_k^{\pm}$ of Proposition~\ref{prop1}. We establish this as follows. 
The two-dimensional vector variable and Dirac-Dunkl operator~\eqref{tilde}
generate another realization of the Lie superalgebra $\mathfrak{osp}(1|2)$. Its Scasimir element, similar to~\eqref{Scasi}, is given by
\[
\tilde\Gamma+1
= \frac12  [\tilde\uD, \tilde\ux ] -\frac12  = \frac12  [\cD_u , u ] + \frac12  [\cD_v , v ] +e_ue_v( v \cD_u-u \cD_v + \frac12  [\cD_u , v ] - \frac12  [\cD_v , u ] )\rlap{\,.}
\]
By means of the commutation relations in Table~\ref{tab:2} we find the explicit form
\[
\tilde\Gamma+1
=\frac12+  \kappa (g_{12}+g_{23}+g_{31}) - e_ue_v  (u \cD_v - v \cD_u)\rlap{\,.}
\]
Comparing with expression~\eqref{Ouv} we observe that 
$\tilde\Gamma+1 =- e_ue_v  O_{uv}$,
and
hence $O_0 = -\ii e_ue_v(\tilde\Gamma+1)$.
Similar to~\eqref{eigGamma}, now using~\eqref{tilderel}
and $ \varphi_k^{\pm}\in\mathcal{M}_k(\mathbb{R}^2,\mathbb{C}^2) = \ker\tilde\uD \cap (\cP_k(\mathbb{R}^2)\otimes \mathbb{C}^2)$ we find
\begin{align*}
(\tilde\Gamma+1) \varphi_k^{\pm}(u,v) &= \frac12(  [\tilde\uD, \tilde\ux ] -1)\varphi_k^{\pm}(u,v) 
\ = \frac12(  \tilde\uD \,\tilde\ux   -1)\varphi_k^{\pm}(u,v) \\*
&= \frac12(  \{\tilde\uD, \tilde\ux \}  -1)\varphi_k^{\pm}(u,v) = \frac12( 2\tilde\mE +2 + 6\kappa  -1)\varphi_k^{\pm}(u,v)
\end{align*}
which, as the Euler operator $\tilde\mE = u\partial_u + v\partial_v$ measures the degree of a polynomial in $u$ and $v$, gives
\[
(\tilde\Gamma+1) \varphi_k^{\pm}(u,v) =\Big( k+\frac12 +3\kappa\Big) \varphi_k^{\pm}(u,v)\rlap{\,.}
\]
Using  $-i e_ue_v \chi^{\pm} = \pm \chi^{\pm}$, which is readily verified using the Pauli matrices~\eqref{pauli}, the action of $O_0$ on $\varphi_k^{\pm}$ then follows to be
\[
O_0 \varphi_k^{\pm}(u,v)=\pm\Big( k+\frac12 +3\kappa\Big)	\varphi_k^{\pm}(u,v)\rlap{\,.}
\]
Since $O_0$ commutes with $\tilde\uD$, $\tilde\ux$ and $e_w w$ we also have, by definition of $\psi_{n,k}^{\pm}$,
\[
O_{0} \psi_{n,k}^{\pm} = \pm  \Big( k+\frac12 +3\kappa\Big) \psi_{n,k}^{\pm}\rlap{\,.}
\]
Finally, as  $O_{123}=(\Gamma+1)e_ue_ve_w$ and
$\psi_{n,k}^{\pm}\in\mathcal{M}_n(\mathbb{R}^3,\mathbb{C}^2)$, by~\eqref{eigGamma} we find the action
\[
O_{123}\psi_{n,k}^{\pm} = \ii (n+1+3\kappa)\psi_{n,k}^{\pm}\rlap{\,.}
\]

To conclude, we consider the action of the $\mathrm{S}_3$ realization on a spinor $\psi_{n,k}^{\pm}$. Using $G_{12}\varphi_k^{\pm}=(-1)^k \varphi_k^{\pm}$, the expressions~\eqref{Gijuvw} and the fact that $G_{12}$ anticommutes with $\tilde{\ux}$ and $\tilde\uD$, we find
\[
G_{12} \psi_{n,k}^{\pm} = (-1)^{n-k}(-1)^k \psi_{n,k}^{\mp} = (-1)^{n} \psi_{n,k}^{\mp}\rlap{\,.}
\] 
Similarly, using now $G_{23}\varphi_k^+=(-1)^k\omega^{\pm(1-k)} \varphi_k^+$ we have
\[
G_{23} \psi_{n,k}^{\pm} = (-1)^{n} \omega^{\pm(1-k)} \psi_{n,k}^{\mp}\rlap{\,.}
\]
This shows, up to rescaling, the correspondence of $\psi_{n,k}^{\pm}$ with the vector $w_k^{\pm}$ of Proposition~\ref{prop1}. 

The abstract inner product on the unitary representation (see section~\ref{sec:unitary}) can now also be realized explicitly. An integral formulation follows by combining the inner product on the spinor space $\mathbb{C}^2$ with the inner product on the unit sphere for Dunkl harmonics~\cite{dunkl2014orthogonal}
\[
\langle \Phi_1, \Phi_2 \rangle = \int_{S^2} ( \Phi_1^{\dagger} \cdot \Phi_2) \, h_{\kappa}^2(u,v,w) \,\mathrm{d}u\mathrm{d}v\mathrm{d}w\rlap{\,,}
\]
where $h_{\kappa}(u,v,w)$ is the $\mathrm{S}_3$ invariant weight function~\cite{dunkl2014orthogonal}
\[
h_{\kappa}(u,v,w) = |u|^{\kappa} | (u^2-3v^2)/4 |^{\kappa}\rlap{\,.}
\]
Using this inner product, the polynomial $\psi_{n,k}^{\pm}$ given by~\eqref{Psiexp} can be normalized to a wavefunction corresponding precisely to the normed vector $w_k^{\pm}$ of Proposition~\ref{prop1}. The orthogonality can be verified by means of the orthogonality relation of the Jacobi polynomials~\cite{Koekoek}.

\section{Conclusion} 
\label{sec:concl}

We presented the symmetry algebra generated by the total angular momentum operators, 
appearing as constants of motion of the $\mathrm{S}_3$ Dunkl Dirac equation. The latter arises as a deformation of the Dirac equation by using Dunkl operators instead of partial derivatives as momentum operators. This corresponds to the addition of a specific potential term to the Dirac Hamiltonian. 
The Dunkl total angular momentum algebra is a one-parameter deformation of the Lie algebra $\mathfrak{so}(3)$ involving reflections. 
We have classified all finite-dimensional irreducible representations of this algebra and we have determined the conditions for the representations to be unitarizable. 
Among the obtained classes of irreducible representations of the symmetry algebra, there is one class of unitary representations for arbitrary positive parameter value. 
This last class admits a natural realization by means of Dunkl monogenics, for which we constructed an explicit basis.

The current results on the symmetry algebra remain to hold when additional potential terms are added to the Hamiltonian. Indeed, the Dunkl total angular momentum operators also commute with functions of the vector variable $\underline{x}=e_1  x_1 +  e_2 x_2 + e_3  x_3$, and thus with a spherically symmetric potential as $\ux^2 = |x|^2$. Furthermore, one may add a deformed spin-orbit interaction term of the form~\eqref{Gamma} and retain the Dunkl total angular momentum components as conserved quantities.


In future work we aim to elevate the setting of the current paper in two directions. On the one hand, one can consider the $N$-dimensional case where the reflection group associated to the Dunkl operator is the symmetric group $\mathrm{S}_N$. On the other hand, it would be interesting to consider more involved root systems (as was done for the type $B_3$ in \cite{GLV}), first in three dimensions and then also in higher dimensions. We look forward to tackle these problems using the insights obtained here.

\section*{Acknowledgments}

	The research of HDB is supported by the Fund for Scientific
	Research-Flanders (FWO-V), project ``Construction of algebra realizations
	using Dirac-operators'', grant G.0116.13N.

\begin{figure}[!htbp]
	\centering
	\includegraphics{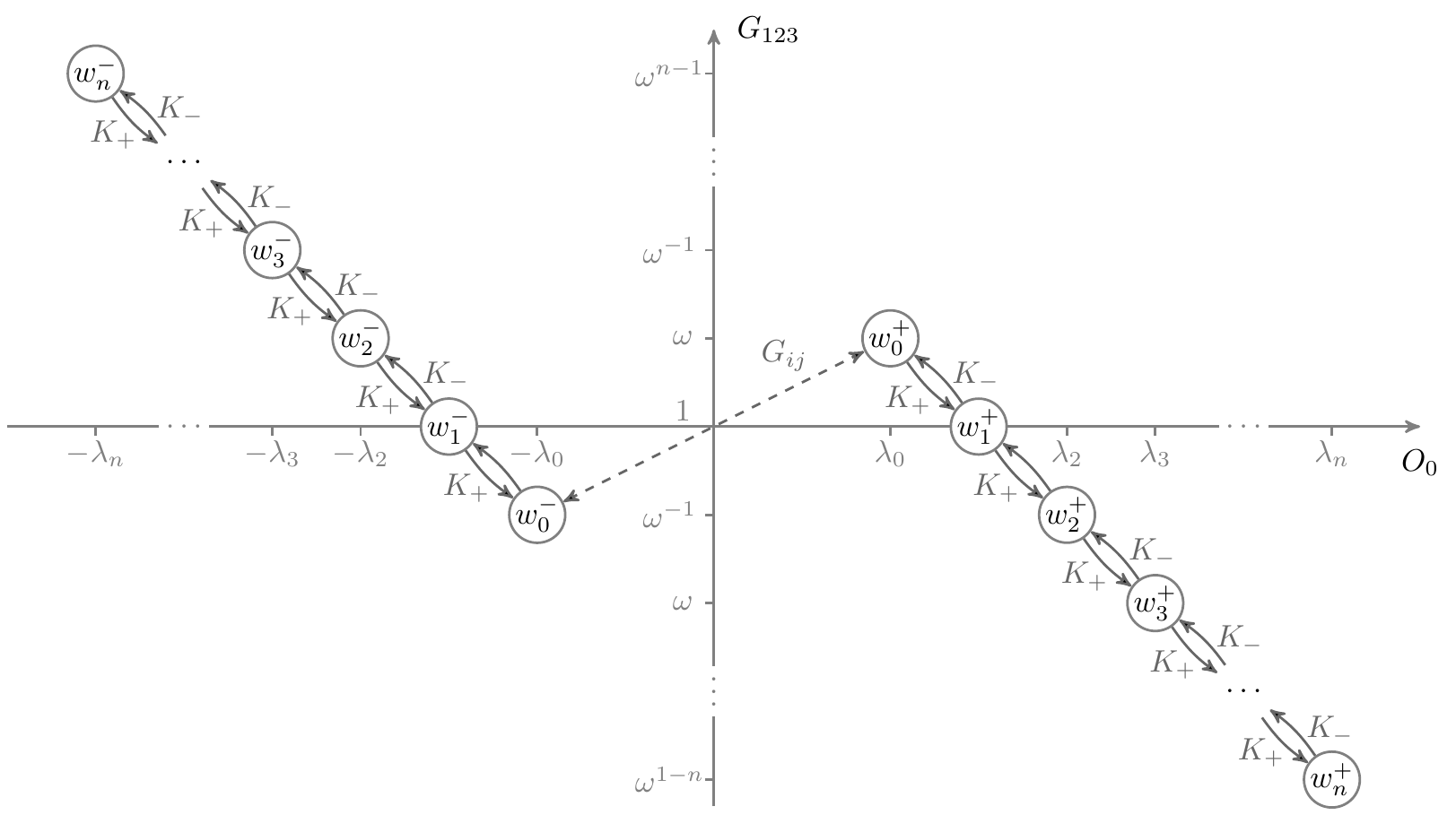}
	\caption{Graphical representation of the basis vectors according to their eigenvalues for $O_0$ and $G_{123}$.
		On the horizontal axis, the shorthand notation $O_0w_k^{\pm} = \pm \lambda_k w_k^{\pm}$ is used, and on the vertical axis the three values $1,\omega,\omega^{-1}$ are repeated periodically. There are two main actions: 1)
		The arrows represent the actions of $K_+$ and $K_-$ through which one moves between the vectors in one half of the vector space. 
		2) In this picture, the action of an odd element of $\mathrm{S}_3$ 
		corresponds to a reflection through the origin, as illustrated for $w_0^+$ and $w_0^-$ by the dashed line. The action of $O_{\pm}$ is a combination of the two main actions.}
	\label{FigRep}
\end{figure}
\end{document}